\documentclass[twocolumn]{autart}    

\usepackage{graphicx} 
\usepackage{xcolor}
\usepackage{cite}
\usepackage{graphicx}
\usepackage{subfigure}
\let\theoremstyle\relax
\usepackage{amsmath,amsthm,amssymb,amsfonts}

\usepackage{algorithm} 
\usepackage{algorithmic}
\usepackage{comment}
\usepackage{xurl}
\usepackage{natbib}
\usepackage{booktabs}
\usepackage{wrapfig}

\newcommand{\tr}{\mathsf{T}}
\newcommand{\poly}{\operatorname{poly}}

\newtheorem{theorem}{Theorem}

\newtheorem{remark}{Remark}[section]
\newtheorem{lemma}{Lemma}
\newtheorem{proposition}{Proposition}[section]
\newtheorem{corollary}{Corollary}[section]

\newtheorem{example}{Example}

\newtheorem{assumption}{Assumption}

\newcommand{\trace}{\operatorname{tr}}

\newcommand{\Ric}{\operatorname{Ric}}

\definecolor{moccasin}{rgb}{0.98, 0.92, 0.84}

\begin{document}

\begin{frontmatter}

\title{
Regret Guarantees for Model-Free Cooperative Filtering under Asynchronous Observations 
\vspace{-5mm}\thanksref{footnoteinfo}} 
\thanks[footnoteinfo]{This work is supported by NSF CMMI 2320697 and NSF CAREER 2340713. Jiachen Qian and Yang Zheng are with the Department of Electrical and
Computer Engineering, University of California, San Diego, 
CA 92093 USA}

\author[UCSD]{Jiachen Qian}\ead{jiq012@ucsd.edu},    
\author[UCSD]{Yang Zheng}\ead{ zhengy@ucsd.edu}               

\address[UCSD]{Department of Electrical and
Computer Engineering, University of California San Diego \vspace{-3mm}}  


\begin{abstract}
Predicting the output of a dynamical system from streaming data is fundamental to real-time feedback control and decision-making. 
We first derive an autoregressive representation that relates future local outputs to \textit{asynchronous} past outputs. Building on this structure, we propose an online least-squares algorithm to learn this autoregressive model for real-time prediction. We then establish a regret bound of $\mathcal{O}(\log^3 N)$ relative to the optimal model-based predictor, which holds for marginally stable systems. Moreover, we provide a sufficient condition characterized via a symplectic matrix, under which the proposed cooperative online learning method provably outperforms the optimal model-based predictor that relies solely on local observations.
From a technical standpoint, our analysis exploits the orthogonality of the innovation process under asynchronous data structure and the persistent excitation of the Gram matrix despite delay-induced asymmetries. 
Overall, these results offer both theoretical guarantees and practical algorithms for model-free cooperative prediction with asynchronous observations, thereby enriching the theory of online learning for dynamical systems. \vspace{-3mm}
\end{abstract}

\end{frontmatter}

\section{Introduction}
\vspace{-3mm}
Online prediction of dynamical systems is a
fundamental problem in many areas \cite{kalmanfilter, box2015time}, with applications ranging from control systems \cite{Andersonoptimal}, robotics \cite{barfoot2024state} to natural language processing \cite{belanger2015linear} and computer vision \cite{coskun2017long}.  
For linear Gaussian stochastic systems, 
the celebrated Kalman filter \cite{kalmanfilter} provides the optimal prediction of system states among all possible filters in the sense of minimal mean square error. Many theoretical results on the Kalman filter, including stability and asymptotic performance, have been well-established in the past \cite{Andersonoptimal, kailath2000linear}. 

While the Kalman filter serves as the foundation and benchmark for prediction with a single centralized data stream, modern large-scale systems pose new challenges that go beyond this classical setting. In applications such as field temperature monitoring \cite{jackson2016novel}, traffic flow control \cite{li2017dynamical}, and power system estimation \cite{xie2012fully}, data are often collected from heterogeneous sources, and the prediction process itself is increasingly decentralized \cite{hashemipour364} or distributed \cite{Olfati4434303}. In such settings, traditional single-source prediction methods are often inadequate.
Consequently, recent works have explored cooperative prediction and information fusion techniques that leverage correlations among different information sources to improve prediction accuracy \cite{hall1997introduction, ren2023effects}.
A key technique in this context is to design~appropriate weights, either scalar or matrix-valued, to effectively fuse the available observations~\cite{niehsen2002information, battistelli2014kullback}. 

\vspace{-3mm}

However, fusing information from different sources in real time introduces two new challenges: 1) information delay induced by signal processing or network communication \cite{wittenmark1995timing, shi2009packetdelay}, and 2) the lack of prior knowledge of the underlying correlation among different sources. 
Temporal asynchrony across sources poses significant difficulties for accurate prediction. 
While techniques such as system augmentation and classical Kalman filter theory allow for deriving the closed-form optimal predictor with delayed information  \cite{Andersonoptimal, lu2005kalman}, these methods often incur substantial computational overhead \cite{larsen1998incorporation} and are further complicated by practical issues such as packet loss  \cite{gu2021kalman}. 
Moreover, in real-world applications such as complex chemistry processes \cite{wu2019machine} or traffic flow \cite{li2017dynamical}, the underlying system dynamics are either unavailable or difficult to identify accurately. In these model-free scenarios, there is no established guideline for leveraging delayed multi-source information for cooperative filtering. 

In this paper, we develop a model-free cooperative prediction algorithm to effectively integrate both local and delayed \textit{external} observations. 
In particular, we consider a linear stochastic system 
\begin{equation}\label{eq: LinearSystem}
\begin{aligned}
x_{k+1}=Ax_k + \omega_k,\;\; y_{k}=Cx_k + v_{k}, \; k = 0,1,2,\ldots 
\end{aligned}
\end{equation}
where $x_k\in\mathbb{R}^n$ is the inherent state vector of the system and $y_{k}\in\mathbb{R}^{m}$ is the measurement output vector, $\omega_k$ and  $v_{k}$ are the process noise and observation noise, respectively. In addition to the local measurement $y_k$, we consider the presence of an external information source, modeled as   
\begin{equation} \label{eq:external-source}
y^{\mathrm{e}}_k=C^{\mathrm{e}}x_k+v^{\mathrm{e}}_k,
\end{equation}
with noise $v_k^{\mathrm{e}} \in \mathbb{R}^{\tilde{m}}$, which may be received via communication or retrieved from the cloud. 
Let $Y_{k_1:k_2}\triangleq\begin{bmatrix}
    y_{k_1},\dots,y_{k_2}
\end{bmatrix},\; Y^{\mathrm{e}}_{k_1:k_2}\triangleq\begin{bmatrix}y^{\mathrm{e}}_{k_1},\dots,y^{\mathrm{e}}_{k_2}\end{bmatrix}$ denote the collection of local and external observations, respectively, from time step $k_1$ to $k_2$. 
We aim to design a cooperative predictor that uses local past observations $Y_{0:k}$ and \textit{delayed} external observations $Y^{\mathrm{e}}_{0:k-d}$, i.e., $\tilde{y}_{k+1}=f\left(Y_{0:k},Y^{\mathrm{e}}_{0:k-d}\right)$, where $d$ is the delay in receiving the external observations. 
One key challenge lies in the {model-free} setting:  
\textit{with no knowledge of the system model $A,C,C^{\mathrm{e}}$ and noise statistics, how can we fully exploit the information in $Y_{0:k}\bigcup Y^{\mathrm{e}}_{0:,k-d}$ to predict $y_{k+1}$?}

\subsection{Our contributions} In this paper, we make four contributions to address this challenge.  

\vspace{-2.5mm}

\begin{itemize}
    \item 
{\bf Autoregressive analysis for multi-source asynchronous observations.} With conditional distribution theory, we first derive the optimal model-based cooperative predictor with partial time delays (Proposition~\ref{theorem:optimal-delayed-filter}). 
Building on this, we construct an autoregressive model that links past delayed outputs to the future output. 
While the delay in external information introduces asymmetry and leads to non-identical innovation dynamics, we prove that this does not compromise the \textit{orthogonality} of the innovations (Theorem~\ref{thm: property}). This preserved {orthogonality} is key to ensuring a logarithmic regret for our online algorithm~next. 

\item {\bf Model-free cooperative filtering with logarithmic regret.} Based on the new auto-regressive model, we propose an online least-squares learning algorithm that leverages multi-source information for model-free cooperative prediction. 
With new insights into the statistical structure of delayed multi-source observations, we establish a regret bound of $\mathcal{O}(\log^3 N)$ for our online algorithm, measuring against the model-based optimal delayed predictor, where $N$ denotes the length of the whole time horizon. This logarithmic regret bound holds for marginal stable systems $\rho(A)=1$ (Theorem~\ref{thm: regret}). 
While the general idea of least-squares learning has been exploited in recent results on model-free Kalman filtering \cite{tsiamis9894660, rashidinejad2020slip,qian2025logarithmic}, our logarithmic regret requires solving a new technical challenge related to asymmetry induced by asynchronous observations. 
\item {\bf Fundamental improvement with asynchronous observations.} 
According to classical model-based information fusion theory, the availability of external data does not necessarily guarantee improved prediction accuracy. We introduce an additional condition based on the symplectic matrix and establish a precise condition under which external delayed information leads to a fundamental performance improvement for online prediction. We then show that, with high probability and for a sufficiently large time horizon $N$, the proposed online cooperative method outperforms the optimal model-based predictor that relies solely on local past outputs (Theorem~\ref{thm: improvement}, Corollary~\ref{coro: improvement}). These results 
demonstrate that external, delayed outputs can indeed improve online prediction performance.

\item {\bf Techniques for Handling asymmetry induced by asynchronous observations.} 
Due to delayed external observations, the autoregressive model exhibits asymmetry across different information sources, making standard analytical techniques in the literature \cite{tsiamis9894660, rashidinejad2020slip, hazan2017learning,qian2025logarithmic} inapplicable. 
We develop new analytical tools to deal with this asymmetry. In particular, we establish a sufficient condition for the persistent excitation of an asymmetric Gram matrix, parameterized by the delay $d$. For any sufficiently large $k\ge\poly\left(d,\log\left(\frac{1}{\delta}\right)\right)$, we prove that the Gram matrix is uniformly persistently exciting with high probability $1-\delta$, i.e., its smallest eigenvalue grows at least linearly with time $k$. 

\end{itemize}

\vspace{-1mm}
\subsection{Related works}

\vspace{-2mm}

The most relevant works can be categorized into two main parts.

\vspace{-2mm}
\begin{itemize}
    \item 
{\bf Cooperative estimation and prediction with network-induced time delay.} 
It is well recognized that using multi-source data improves online prediction performance \cite{hall1997introduction, wu2011wireless}.
  Still, the time delay induced by communication remains a critical challenge. Classical work focused on describing the effect of time delay on the estimation process \cite{alexander1991state, wittenmark1995timing}. Subsequent works aim to design the optimal filter with delayed observations \cite{shidistributed5200489} or suboptimal filters with reduced computational complexity \cite{larsen1998incorporation, shan2015delayed}. The closest work to our setup is \cite{shidistributed5200489}. We provide an improved analysis of innovation orthogonality in a more concise form than that in \cite{shidistributed5200489}. We further quantify the improvement in prediction by external delayed information in terms of the delay step $d$.
\item 
{\bf Model-free online prediction. }  
A central theoretical question in this line is whether an online learning algorithm can achieve sublinear \textit{regret} against the model-based optimal predictor \cite{hazan2017learning,ghai2020no,tsiamis9894660,qian2025model,qian2025logarithmic}. For linear Gaussian stochastic systems, \cite{tsiamis9894660} demonstrated that a truncated regression model combined with an online least squares algorithm can ensure \textit{logarithmic regret}. With a spectral filtering technique, which leverages low rank approximation to mitigate the effect of long-term dependence, \cite{hazan2017learning} established a sublinear regret bound $\sqrt{N}$ for general linear stochastic systems,  and \cite{rashidinejad2020slip} further provided an improved logarithm regret bound $\log^{11} N$ for Gaussian stochastic systems.  { Our recent work \cite{qian2025model} introduces an exponential forgetting strategy to mitigate overfitting arising from the unbalanced autoregressive structure in online prediction, while still preserving logarithmic regret. In \cite{qian2025logarithmic}, we further extend the online regression framework to multi-step prediction and establish that, relative to the optimal multi-step Kalman predictor, the regret remains logarithmic in the time horizon, with the constant factor scaling polynomially with the prediction horizon. }
Despite these recent advances, model-free cooperative filtering is much less well developed, 
and theoretical guarantees for model-free cooperative filtering remain largely open. Our work establishes the first logarithmic regret for model-free cooperative prediction using asynchronous observations.     
\end{itemize}

\subsection{Notations}
We write $\succ$ for the Loewner order on the positive semidefinite cone.
The norms $\|\cdot\|_2$, $\|\cdot\|_F$, and $\|\cdot\|_1$ denote the spectral (operator) norm, Frobenius norm, and entrywise $\ell_1$ norm, respectively.
For a sub-Gaussian random variable $X$, $\|X\|_{\psi_2}$ denotes its $\psi_2$ norm \citep[Definition~2.5.6]{vershynin2018high}.
We write $\rho(A)$ for the spectral radius of a matrix $A$, and $\operatorname{poly}(x)$ for a polynomial in the components of $x$.
Expectation and probability are denoted by $\mathbb{E}[\cdot]$ and $\mathbb{P}[\cdot]$, respectively.
The notation $O(f(x))$ and $\tilde O(f(x))$ have their usual meanings, where $\tilde O(f(x))$ suppresses polylogarithmic factors.
The Kronecker delta is $\delta_{kl}=1$ if $k=l$ and $\delta_{kl}=0$ otherwise. 
Finally, we define the Riccati operator
$
\Ric(A,C,Q,R,P)\triangleq \;APA^\tr+Q-APC^\tr\left(CPC^\tr+R\right)^{-1}CPA^\tr.
$
\section{Preliminaries}

\vspace{-1mm}

\subsection{System model and centralized Kalman filter} 

\vspace{-2mm}

We consider the linear stochastic system \eqref{eq: LinearSystem} in the presence of an external information source \eqref{eq:external-source}. 
The sequences $\left\{\omega_k\right\}^{\infty}_{k=0}$, $\left\{v_{k}\right\}^{\infty}_{k=0}$ and $\left\{v_{k}^{\mathrm{e}}\right\}^{\infty}_{k=0}$ are {mutually uncorrelated white Gaussian noises}. Their covariance matrices are denoted as $Q$, $R$, and $R^{\mathrm{e}}$, respectively, and they are positive definite.
At time $k\geq d$, we have the local observations $y_0, \ldots, y_k$ and delayed external observations $y_0^{\mathrm{e}}, \ldots, y_{k-d}^{\mathrm{e}}$. Based on the observations, we aim to predict the next observation $\hat{y}_{k+1}$. 

If we only have the local observations $Y_{0:k}$, the classical theory minimizes the mean-square error as:  
\begin{equation}\label{MMSEProblem-main}
    \hat{y}_{k+1} \triangleq \arg \min _{z \in \mathcal{F}_{k}} \mathbb{E}\left[\left\|y_{k+1}-z\right\|_{2}^{2} \mid \mathcal{F}_{k}\right], 
\end{equation}
where $\mathcal{F}_{k} \triangleq \sigma\left(y_{0}, \ldots, y_{k}\right)$  denotes the filtration generated by the local observations  $y_{0}, \ldots, y_{k}$. 
It is now well-known that the solution to \eqref{MMSEProblem-main} takes a recursive form, i.e., the \textit{Kalman filter} \cite{kalmanfilter}, 
\begin{equation}\label{eq: localOptimal-main}
\begin{aligned}
\hat{y}_{k}  =C \hat{x}_{k},\qquad\hat{x}_{k+1} =A \hat{x}_{k}+L_k\left( y_{k}-\hat{y}_{k}\right),
\end{aligned}
\end{equation}
where $L_k=AP_kC^\tr\left(CP_kC^\tr+R\right)^{-1}$ with $P_k$ satisfying the Riccati recursion
\begin{equation}\label{eq: PartialRicRecursion-main}
    P_{k+1}=\Ric(A,C,Q,R,P_k). 
\end{equation}
The prediction error's covariance is 
$$\mathbb{E}\left\{(y_{k+1}-\hat{y}_{k+1})(y_{k+1}-\hat{y}_{k+1})^\tr\right\}
=CP_{k+1}C^\tr+R.
$$
\vspace{-7mm}

The solutions in \eqref{MMSEProblem-main}-\eqref{eq: PartialRicRecursion-main} are standard, but without using the external observations from \eqref{eq:external-source}. Intuitively, 
 access to this additional information can enhance prediction performance, particularly by reducing the covariance of the prediction error. The simplest case occurs when there is no delay in receiving the external observations, i.e., $d = 0$. In this case, we let $y^{\mathrm{c}}_k\triangleq\left[y_k^\tr,y_k^{\mathrm{e}\tr}\right]^\tr$ denote the centralized observation at each time step $k$; accordingly, we denote $\bar{C}=\left[C^\tr,C^{\mathrm{e}\tr}\right]^\tr$ and $\bar{R}=\operatorname{diag}\left\{R,R^{\mathrm{e}}\right\}$. 
With no delay, the filtration in \eqref{MMSEProblem-main} will be $\mathcal{F}_{k} \triangleq \sigma(y_{0}^{\mathrm{c}}, \ldots, y_{k}^{\mathrm{c}})$, where each $y_k^{\mathrm{c}}$ is the combination of $y_{k}$ and $y_{k}^{\mathrm{e}}$. Then, the corresponding optimal Kalman filter takes the same recursive form: 
\begin{equation}\label{eq: globalOptimal}
\begin{aligned}
\bar{y}_{k}  =C \bar{x}_{k},\qquad\bar{x}_{k+1} =A \bar{x}_{k}+\bar{L}_k\left( y_{k}^{\mathrm{c}}-\bar{C}\bar{x}_{k}\right),
\end{aligned}
\end{equation}
where the corresponding Kalman gain is $\bar{L}_{k}=A\bar{P}_k\bar{C}^{\tr}\left(\bar{C}\bar{P}_k\bar{C}^\tr+\bar{R}\right)^{-1}$ with $\bar{P}_k$ satisfying the same Riccati recursion \eqref{eq: PartialRicRecursion-main} in which $C$ and $R$ are replaced by $\bar{C}$ and $\bar{R}$ respectively. For \eqref{eq: globalOptimal}, the prediction error's covariance is 
$\mathbb{E}\left\{(y_{k}-\bar{y}_{k})(y_{k}-\bar{y}_{k})^\tr\right\}=C\bar{P}_kC^\tr+R$. 

Indeed, we can show that $P_k \succeq \bar{P}_k,\;\;\forall k>0$ (see, e.g., \cite{ren2023effects}), thus, accessing the external observation reduces the covariance of the prediction error. To illustrate this, we present a simple example below.

\begin{example} \label{example:comparison}
    Consider a second-order system with noise statistics below
    \begin{equation}\label{eq: examplePara}
    \begin{aligned}
        A&=\begin{bmatrix}
    0.2&0.8\\
    0.4&0.6
\end{bmatrix}, \quad \bar{C}=\begin{bmatrix}
    C\\C^{\mathrm{e}}
\end{bmatrix}=
\begin{bmatrix}
1&0\\0&1
\end{bmatrix},\\  
Q &= \begin{bmatrix}
    1&0\\0&1
\end{bmatrix},
\qquad \bar{R} =\begin{bmatrix}
R\\&R^{\mathrm{e}}
\end{bmatrix}=\begin{bmatrix}
    1&0\\0&1
\end{bmatrix}.
    \end{aligned} 
\end{equation}
This system has a two-dimensional state space. The local sensor only measures the first state, while an external sensor measures the second state. As guaranteed by the centralized Kalman filter, we can reduce the covariance of the prediction error by using both measurements without delay. This is indeed the case: with only the local sensor, the steady-state error covariance is { 3.10} while this is reduced to { 2.40} with both sensors. We illustrate the error covariances in Fig.~\ref{fig: Motivation}.  \hfill $\square$
\end{example}

\begin{figure}[t]
    \centering
    \subfigure[Local vs. Cooperative filtering.]{\includegraphics[width=0.4\textwidth]{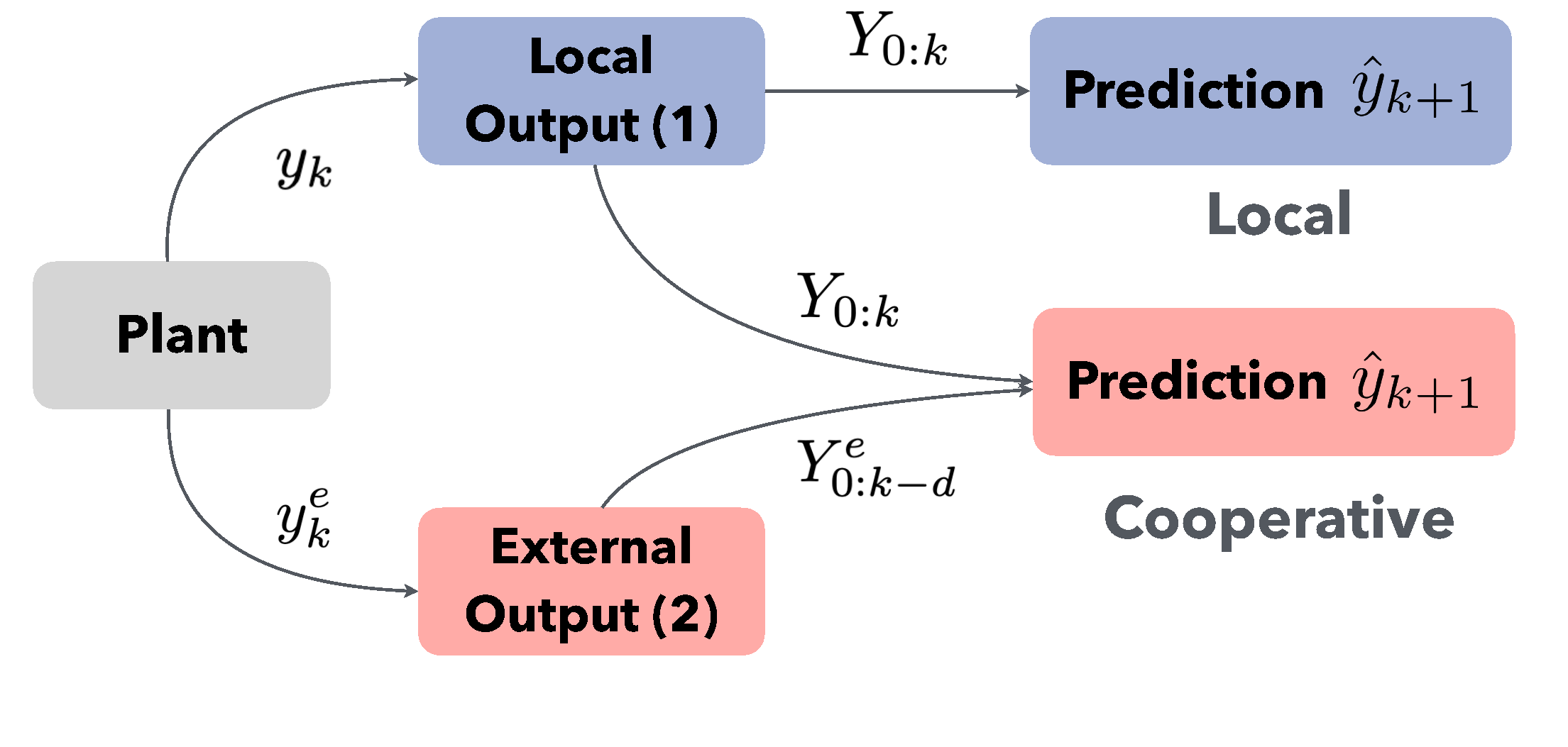}
\label{figMotivation}}
\hspace{6mm}
     \subfigure[Numerical comparison using Example~\ref{example:comparison}.]
{\includegraphics[width=0.4\textwidth]{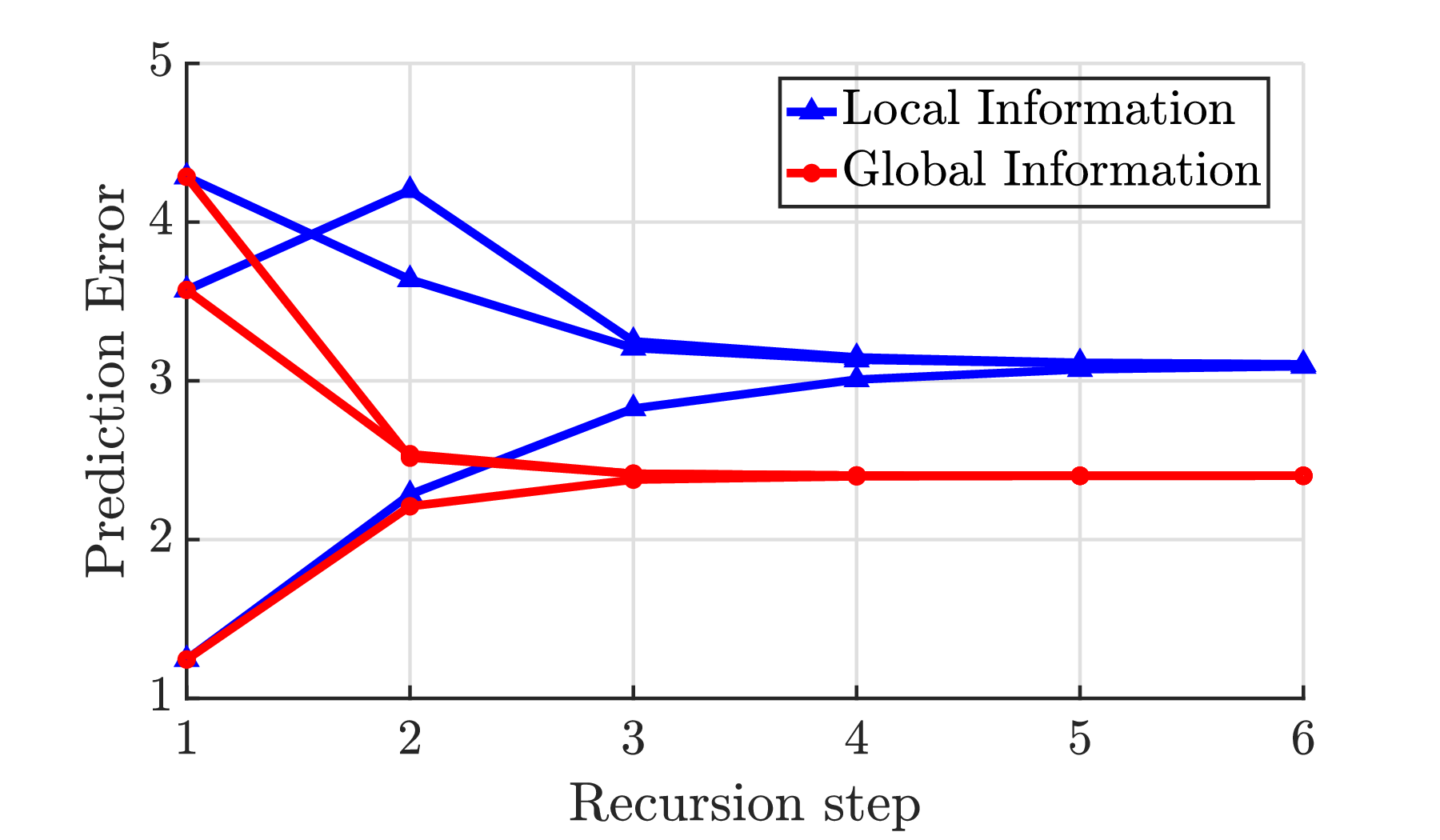}
\label{fig: Motivation}}
    \caption{Comparison between a standard local filter and our cooperative filter with delayed external observations. (a) Schematic illustration where our cooperative filter utilizes delayed external observations in \eqref{eq:external-source}; (b) Numerical results in which accessing additional observations reduces the prediction error's covariance (see Example~\ref{example:comparison}). } 
\end{figure}

\vspace{-1.5mm}
\subsection{Online learning and regret}

\vspace{-1mm}

The standard Kalman filter theory requires the exact system model and noise statistics, as shown in \eqref{eq: localOptimal-main} and \eqref{eq: PartialRicRecursion-main}. In this paper, we focus on \textit{model-free} online learning algorithms to predict observations without requiring knowledge of the system matrices or noise covariances.

Following \cite{tsiamis9894660,rashidinejad2020slip,qian2025model}, we quantify the performance of the online prediction in terms of the \textit{regret} measured against the Kalman filter \eqref{eq: localOptimal-main} that has full system knowledge. 
If we only use local measurement, the goal is to design an online algorithm $\tilde{y}_k=f_k(y_0,\dots,y_{k-1})$ to minimize the following regret
\begin{equation}\label{regret}
    \mathcal{R}_{N} \triangleq \sum_{k=1}^{N}\left\|y_{k}-\tilde{y}_{k}\right\|^{2}-\sum_{k=1}^{N}\left\|y_{k}-\hat{y}_{k}\right\|^{2},
\end{equation}
where $\hat{y}_{k}$ is the steady-state  Kalman filter's prediction from \eqref{eq: localOptimal-main}.  
Due to the exponential convergence of the Kalman filter, it is sufficient to compare the online algorithm with the steady-state predictor $\hat{y}_k$. 
Recent advances in \cite{tsiamis9894660,rashidinejad2020slip,qian2025model} have shown that it is possible to obtain \textit{logarithmic} regrets for the standard Kalman filter.  
One key observation is that the Kalman filter \eqref{eq: localOptimal-main} can be utilized to build up an \textit{auto-regressive} structure for $y_k$ after expanding the past $p$ observations: %
\begin{equation}\label{regression}
    y_{k}=G_{p} Z_{k, p}+C(A-L C)^{p} \hat{x}_{k-p}+r_{k},
\end{equation}
where $r_k = y_k - \hat{y}_k$ denotes the innovation, and $Z_{k, p}$ denotes the past $p$ observations at time $k$, i.e., 
$
Z_{k, p} \triangleq\begin{bmatrix}
    y_{k-p}^{\tr}& \ldots & y_{k-1}^{\tr}
\end{bmatrix}^\tr,
$
and $G_{p}$ denotes the Markov parameters of the closed-loop Kalman filter: 
\begin{equation} \label{eq:regressor}
G_{p} \triangleq\begin{bmatrix}
C(A-L C)^{p-1} L, & \cdots, & C L
\end{bmatrix} \in \mathbb{R}^{m \times pm},
\end{equation}
where $L$ is the steady-state Kalman gain from the Riccati equation induced by the steady-state form of \eqref{eq: PartialRicRecursion-main}.

\subsection{Cooperative filtering with asynchronous observations and its regret}

\vspace{-2mm}

We aim to exploit the external observation from \eqref{eq:external-source} while accounting for practical asynchrony modeled by a time delay $d$. In this setting, neither the centralized Kalman filter \eqref{eq: globalOptimal} nor the standard regret measure \eqref{regret} is directly applicable. First, the information of $Y_{0:k}^{\mathrm{e}}$ is not fully available at time step $k$; in contrast, we can only obtain delayed observations $Y_{0:k-d}^{\mathrm{e}}$. Second, the regret benchmark in \eqref{regret} might be too weak as it ignores the contribution of external observations. 

To address both limitations, we develop a model-free cooperative filtering algorithm with asynchronous observations. Our cooperative predictor 
integrates local past observations $Y_{0:k}$ and \textit{delayed} external observations $Y^{\mathrm{e}}_{0:k-d}$, i.e., $\tilde{y}_{k+1}=f\left(Y_{0:k},Y^{\mathrm{e}}_{0:k-d}\right)$. This predictor performs better than the local Kalman filter \eqref{eq: localOptimal-main}. We will develop the optimal model-based cooperative predictor with partial time delays and use it as the benchmark in the regret measure. 
Our model-free prediction algorithm will be based on a new autoregressive model that links past delayed outputs to future observations (Section \ref{section:auto-regressive-model}). With technical tools including persistent excitation of structurally asymmetric Gram matrices and uniform logarithmic bounds for stochastic processes, we further establish a \textit{logarithmic regret} for our model-free cooperative filter (Section \ref{section:algorithm&regret}). Before proceeding to our technical results, we make two assumptions in this work.

\vspace{6pt}
\begin{assumption}\label{asp: Ob-main-text}
   $(A, C)$ and $\left(A,C^{\mathrm{e}}\right)$ are detectable and $A$ is marginally stable, i.e., $\rho(A)\leq 1$. 
\end{assumption}
\vspace{6pt}

\begin{assumption}\label{asp: Diagonal}
    The matrix $A-\bar{L}\bar{C}$ is diagonalizable, where $\bar{L}$ is the steady-state Kalman gain~in~\eqref{eq: globalOptimal}. 
\end{assumption}
The first assumption is standard in the design of optimal filters \cite{Andersonoptimal} to guarantee the existence of the steady-state filtering covariance $P$ and gain $L$. Even without this assumption, we can still apply a linear transformation to \eqref{eq: LinearSystem} and then discuss the filtering performance in the detectable subspace \cite{battistelli2014kullback}.  The second assumption does not affect our filter design but simplifies our regret analysis. Assumption \ref{asp: Diagonal} is also used in the analysis of model-free centralized Kalman filter \cite{tsiamis9894660,rashidinejad2020slip,qian2025model}. Removing Assumption \ref{asp: Diagonal} does not pose a fundamental difficulty, but this will make the regret guarantee less clean.

\vspace{-2mm}

\section{Autoregressive modeling with asynchronous observations} \label{section:auto-regressive-model}

\vspace{-1mm}
\subsection{Optimal predictor with asynchronous observations}
\vspace{-1mm}

Recall that $Y_{0:k}$, $Y_{0:k}^{\mathrm{e}}$, and $Y_{0:k}^{\mathrm{c}}$ denote the local observation sequence, external observation sequence, and centralized observation sequence, respectively, from time $0$ to $k$. We here develop the optimal predictor $\hat{y}_{k+1}$ with local observations $Y_{0:k}$ and delayed external observations $Y_{0:k-d}^{\mathrm{e}}$. 
From standard optimal filtering theory \cite{kalmanfilter, Andersonoptimal}, the minimal-mean-square-error (MMSE) predictor $\hat{y}_{k+1}$ is the expectation of $y_{k+1}$ conditioned on past observations $Y_{0:k}$ and $Y_{0:k-d}^{\mathrm{e}}$, i.e.,
\begin{equation} \label{eq:best-output-estimiation}
\hat{y}_{k+1}\!=\!\mathbb{E}\left\{y_{k+1}\!\mid\! Y_{0:k},Y_{0:k-d}^{\mathrm{e}}\right\}\!=\!C\mathbb{E}\left\{x_{k+1}\!\mid\! Y_{0:k},\!Y_{0:k-d}^{\mathrm{e}}\right\},
\end{equation}
where the last equality uses the zero mean of noise $v_{k+1}$. Thus, it suffices to find the best state estimator $\hat{x}_{k+1}$ based on $Y_{0:k}$ and $Y_{0:k-d}^{\mathrm{e}}$. 

Let $\bar{x}_{k+1}\triangleq \mathbb{E}\left\{x_{k+1}\mid Y_{0:k}^{\mathrm{c}}\right\} $, i.e., optimal Kalman state prediction with centralized observations $Y_{k}^{\mathrm{c}}$ from \eqref{eq: globalOptimal}. 
Thanks to the Markov property, we have 
\begin{equation} \label{eq:two-step-state-estimiation}
\begin{aligned}
    \hat{x}_{k+1}\triangleq&\,\mathbb{E}\left\{x_{k+1}\mid Y_{0:k},Y_{0:k-d}^{\mathrm{e}}\right\}\\
    =&\, \mathbb{E}\left\{x_{k+1}\mid y_{k},\dots,y_{k-d+1},Y_{0:k-d}^{\mathrm{c}}\right\}\\
    =&\,\mathbb{E}\left\{x_{k+1}\mid y_{k},\dots,y_{k-d+1},\bar{x}_{k-d+1}\right\}.
\end{aligned}
\end{equation}
The expression in \eqref{eq:two-step-state-estimiation} informs us that we can separate the estimation $\mathbb{E}\left\{x_{k+1}\mid Y_{0:k}, Y_{0:k-d}^{\mathrm{e}}\right\}$ into two parts: 1) at time step $k$, we obtain the local observation $y_k$ and delayed external observation $y_{k-d}^e$, then we can first obtain the centralized optimal Kalman prediction for step $k-d+1$, i.e., $\bar{x}_{k-d+1}$ as we have the collected observations $Y_{0:k-d}^{\mathrm{c}}$; 2) we then refine the state estimation using local observations $y_{l}, l = k-d+1, \ldots, k$ from $\bar{x}_{k-d+1}$ to get $\hat{x}_{k+1}$. The second step can be realized using standard recursion of the local Kalman filter for $d$ times. 
Based on the discussions, we now present the following result.

\vspace{6pt}

\begin{proposition} \label{theorem:optimal-delayed-filter}
    Consider the linear stochastic system \eqref{eq: LinearSystem} in the presence of an external information source \eqref{eq:external-source} with i.i.d. Gaussian noises. The steady-state MMSE predictor $\hat{x}_{k+1}$ based on the local observations $y_0, \ldots, y_k$ and delayed external observations $y_0^{\mathrm{e}}, \ldots,y_{k-d}^{\mathrm{e}}$ is \vspace{-1mm}
    \begin{equation}\label{eq: steadyPredictor}
    \hat{x}_{k+1}=  \Phi_{d} \bar{x}_{k-d+1} + \sum_{l=1}^{d}\Phi_{d-l}L^{(l)}\;y_{k-d+l}, \vspace{-1mm}
\end{equation}
where $\bar{x}_{k-d+1}\triangleq \mathbb{E}\left\{x_{k-d+1}\mid Y_{0:k-d}^{\mathrm{c}}\right\}$ from \eqref{eq: globalOptimal}. In \eqref{eq: steadyPredictor}, we have $ \Phi_{d-l} = \prod_{i=l+1}^d (A-L^{(i)}C), \Phi_{0}=I,  
   L^{(l)} = AP^{(l)}C^\tr\left(CP^{(l)}C^\tr+R\right)^{-1}$ 
with the recursion of $P^{(l)}$ constructed as
\begin{equation}\label{eq: steadyRecursion}
    P^{(l+1)}=\Ric(A,C,Q,R,P^{(l)}), \quad P^{(1)}=\bar{P},
\end{equation}
where $\bar{P}$ is the unique positive semidefinite solution to the discrete-time Riccati equation
$
\;\bar{P}=\Ric(A,\bar{C},Q,\bar{R},\bar{P}).
$
\end{proposition}

Proposition \ref{theorem:optimal-delayed-filter} presents the optimal model-based predictor under partially delayed information. This result can be viewed as a special case of the more general setting in \cite[Theorem 3.1]{shidistributed5200489}, though the predictor in our case, given by \eqref{eq: steadyPredictor}, is greatly simplified. 
{For completeness, we provide a self-contained and simplified proof in Appendix \ref{Appen: timevarying}.} From \eqref{eq:best-output-estimiation}, it is clear that the optimal local prediction is $\hat{y}_{k+1} = C \hat{x}_{k+1}$ with $\hat{x}_{k+1}$ from \eqref{eq: steadyPredictor}. However, similar to the classical Kalman filter, this predictor \eqref{eq: steadyPredictor} is model-based.

\vspace{-1mm}

\subsection{Auto-regressive model and orthogonality of the innovation process}

\vspace{-1mm}

To develop a model-free prediction algorithm, we here establish an autoregressive model 
that link the future $y_{k+1}$ with the past observations $Y_{k}\bigcup Y_{k-d}^{\mathrm{e}}$. 
With the optimal delayed estimator $\hat{x}_{k+1}$ in  \eqref{eq: steadyPredictor}, we can write the local output at $k+1$ as
\begin{equation} \label{eq:innovation-form}
y_{k+1}=C\hat{x}_{k+1}+r_{k+1},
\end{equation}
where $r_{k+1}$ denotes the optimal prediction error (also called {\it innovation}). This innovation $r_{k+1}= y_{k+1} - C\hat{x}_{k+1}$ represents the amount of new information $y_{k+1}$ brings beyond what is already known from the optimal estimation $\hat{x}_{k+1}$. Substituting the expression of $\hat{x}_{k+1}$ in \eqref{eq: steadyPredictor}, we have
$$
y_{k+1}=C\Phi_{d} \bar{x}_{k-d+1} + C\sum_{l=1}^{d}\Phi_{d-l}L^{(l)}y_{k-d+l}+r_{k+1}. 
$$
Note that the steady-state Kalman filter satisfies 
$$
\bar{x}_{k-d+1}\!\!=\!\!(A-\bar{L}\bar{C})^p \bar{x}_{k-d-p+1}+\!\sum_{l=1}^p (A-\bar{L}\bar{C})^{l-1} \bar{L}y_{k-d+1-l}^{\mathrm{c}}, 
$$ where $\bar{L}$~is~the centralized steady-state Kalman gain in~\eqref{eq: globalOptimal}. 
Then, the innovation form \eqref{eq:innovation-form} leads to an auto-regressive~model 
\begin{equation}\label{eq: regressionModel}
    y_{k+1}\!\!=\!\!G_{p+d}Z_{k+1,p+d}+\!C \Phi_d(A\!-\!\bar{L} \bar{C})^{p} \bar{x}_{k-d-p}\!+r_{k+1} 
\end{equation}
where $Z_{k+1,p+d}\triangleq\begin{bmatrix}
    y_{k}^{\tr}&\cdots&y_{k-d+1}^{\tr}&y_{k-d}^{c\tr}\cdots&y_{k-d-p+1}^{c\tr}
\end{bmatrix}^\tr$ contains the past local observations and delayed external observations. The coefficient matrix $G_{p+d}$ is given as 
$
G_{p+d}\triangleq\begin{bmatrix}
     G_{d}^{(1)},G_p^{(2)}
\end{bmatrix},$ with $G_{d}^{(1)}\!\triangleq\!\begin{bmatrix}
    C L^{(d)}&\! \cdots \!& C \Phi_{d-1}L^{(1)}
\end{bmatrix}$ denotes the weights for the undelayed local observations and $G_p^{(2)}\triangleq \begin{bmatrix}
    C \Phi_d \bar{L} &\cdots &C \Phi_d(A-\bar{L} \bar{C})^{p-1} \bar{L}
\end{bmatrix}$ denotes the weights of the available global observations.
At this stage, we observe that the autoregressive model \eqref{eq: regressionModel} has a form similar to the standard Kalman filter \eqref{regression}. The main difference is that the auto-regressive model \eqref{eq: regressionModel} incorporates both past local observations and delayed external observations. 
The properties of the autoregressive model \eqref{eq: regressionModel} are summarized below. 

\vspace{6pt}
 \begin{theorem}\label{thm: property}
 Consider the auto-regressive model \eqref{eq: regressionModel} for the linear stochastic system \eqref{eq: LinearSystem} with an external information source \eqref{eq:external-source}. The following statements hold: 
 \begin{enumerate}
     \item The innovation process $r_{k}$ is orthogonal, i.e., 
 $$
 \mathbb{E}\left\{r_{k}r_{l}^\tr\right\}= 
  \begin{cases}
      CP^{(d+1)}C^\tr+R  & \text{if} \; k = l, \\
     0  & \text{otherwise}. 
 \end{cases}
 $$
 \vspace{-5pt}
 \item The steady-state variance $\mathbb{E}\left\{r_kr_k^\tr\right\}$ exponentially converges from $C\bar{P}C^\tr+R$ to $CPC^\tr+R$ as $d\to \infty$.
 \vspace{2pt}
 \item The state-transition matrix is exponentially stable, i.e., $\|\Phi_d\|_2^2\leq \tau\rho_0^d$, where 
    $\tau = \frac{\lambda_{\max}(P^{(d+1)})}{\lambda_{\min} (P^{(1)})} \,  \text{and} \, \rho_0=1-\frac{\lambda_{\min} (Q)}{\lambda_{\max}(P^{(d+1)})} \in (0,1).$ 
 \end{enumerate}
 \end{theorem}

\begin{proof}
\textbf{Proof of Part 1: }\label{Appen: pfOrthogonal}
We first make the following notations
$
e_k\triangleq x_k-\hat{x}_k,\; \bar{e}_k\triangleq x_k-\bar{x}_k,
$ to denote the prediction error of the optimal predictors with partially delayed information and complete information, respectively.
We will divide this part into three steps. Without losing generality, we assume $k\ge l$.

\vspace{-10pt}
{\bf Step 1:} The case with $k=l$ can be directly verified with
the expression
\vspace{-10pt}
\begin{equation}\label{eq: totalError}
    e_k\!\!=\!\Phi_d\bar{e}_{k-d}+\!\sum_{i=0}^{d-1}\!\Phi_{d-i-1}w_{k-d+i}-\!\!\sum_{i=0}^{d-1}\!\!\Phi_{d-i-1}\tilde{v}_{k-d+i}
\end{equation}
with $\tilde{v}_{k-d+i}=L^{(i+1)}v_{k-d+i}$. Then, with the mutual unrelatedness among $w_{k}$ and $v_k$, the covariance equation can be verified directly through algebraic manipulation.

\vspace{3pt}
{\bf Step 2:} For the case with $k-l> d$, we first have
$
\mathbb{E}\left\{r_{k}r_{l}^\tr\right\}=C\mathbb{E}\left\{e_ke_l^\tr\right\}C^\tr+C\mathbb{E}\left\{e_kv_{l}^\tr\right\}.
$
For the term $C\mathbb{E}\left\{e_kv_{l}^\tr\right\}$, denote $\bar{A}\triangleq A-\bar{L}\bar{C}$ as the centralized closed-loop error matrix and $s\triangleq k-l-d$ as the lag parameter, we first have
\vspace{-5pt}
\begin{equation}\label{eq: rollout_error}
    \begin{aligned}
    \bar{e}_{k-d}\!=\!\bar{A}^{s}\bar{e}_{l}+\!\sum_{i=1}^{s}\bar{A}^{s-i}w_{l+i-1}-\!\sum_{i=1}^{s}\bar{A}^{s-i}\bar{L}v^{c}_{l+i-1},
\end{aligned}
\vspace{-5pt}
\end{equation}
where $v^{c}_{l}=\left[v_l^\tr,v_l^{\mathrm{e}\tr}\right]^\tr$ is the collection of noise vectors at time step $l$.
Together with \eqref{eq: totalError}, we have
$
C\mathbb{E}\left\{e_kv_{l}^\tr\right\}=-\Phi_d\bar{A}^{s-1}\bar{L}\tilde{R},
$
where $\tilde{R}=\begin{bmatrix}
    R,O_{m\times \tilde{m}}
\end{bmatrix}^\tr\in\mathbb{R}^{(m+\tilde{m})\times m}$ denotes the expectation $\mathbb{E}\left\{v_{l}^{c}v_l^\tr\right\}$.
For the term $C\mathbb{E}\left\{e_ke_l^\tr\right\}C^\tr$, we first have
$
\mathbb{E}\left\{e_ke_l^\tr\right\}=\Phi_d\bar{A}^{s}\mathbb{E}\left\{\bar{e}_le_l^\tr\right\}.
$
While for $\mathbb{E}\left\{\bar{e}_le_l^\tr\right\} $, due to the structural difference between the two error terms, we need first to roll out the error $\bar{e_l}$ and $e_l$ for $d$-step, with some basic calculation, we have 
\vspace{-5pt}
\[
\begin{aligned}
    \mathbb{E}\left\{\bar{e}_le_l^\tr\right\}=&\bar{A}^d \bar{P}\Phi_d^\tr+\sum_{i=0}^{d-1}\bar{A}^{d-i-1}Q\Phi_{d-i-1}^\tr\\&+\sum_{i=0}^{d-1}\bar{A}^{d-i-1}\bar{L}\tilde{R}L^{(i+1)^\tr}\Phi_{d-i-1}^\tr.
\end{aligned}
\vspace{-5pt}
\]
This expression indicates that the structural difference between the two error dynamics makes the cross-covariance more difficult to analyze.
To simplify the above expression, we need the following statement:
for any $i=0,\dots,d-1$, there is 
\vspace{-5pt}
\begin{equation}\label{eq: crossClaim}
    \bar{P}=\bar{A}\bar{P}(A-L^{(i+1)}C)^\tr+Q+\bar{L}\tilde{R}L^{(i+1)\tr}.
    \vspace{-5pt}
\end{equation}
To show this, we substitute $\bar{P}=\bar{A}\bar{P}\bar{A}^\tr+Q+\bar{L}\bar{R}\bar{L}^\tr$ into \eqref{eq: crossClaim} then obtain
\[
\begin{aligned}
    &\bar{P}-\left(\bar{A}\bar{P}(A-L^{(i+1)}C)^\tr+Q+\bar{L}\tilde{R}L^{(i+1)\tr}\right)\\
    =&\bar{A}\bar{P}\bar{C}^\tr L^{(i+1)^\tr}-\bar{L}\tilde{R} L^{(i+1)\tr}-\bar{A}\bar{P}\bar{C}^\tr \bar{L}^\tr+\bar{L}\bar{R}\bar{L}^\tr.
\end{aligned}
\]

With the expression $\bar{L}=A\bar{P}\bar{C}^\tr(\bar{C}\bar{P}\bar{C}^\tr+\bar{R})^{-1}$, we have
\[
\bar{L}\bar{R}\bar{L}^\tr\!-(A\!-\bar{L}\bar{C})\bar{P}\bar{C}^\tr \bar{L}^\tr\!=\!\bar{L}(\bar{C}\bar{P}\bar{C}^\tr\!+\!\bar{R})\bar{L}^\tr\!-\!A\bar{P}\bar{C}^\tr \bar{L}^\tr\!=\!0
\]
and
$
    \bar{A}\bar{P}C^\tr L^{(i+1)^\tr}-\bar{L}\tilde{R} L^{(i+1)\tr}
    =A\bar{P}C^\tr L^{(i+1)^\tr}-\bar{L}(\bar{C}\bar{P}C^\tr+\tilde{R})L^{(i+1)^\tr}=0,
$
where the last equality holds due to 
\vspace{-5pt}
\[
(\bar{C}\bar{P}\bar{C}^\tr+\bar{R})^{-1}(\bar{C}\bar{P}C^\tr+\tilde{R})=\begin{bmatrix}
    I\\0
\end{bmatrix}_{(m+\tilde{m})\times m}.
\vspace{-5pt}\]
The claim \eqref{eq: crossClaim} is thus proved. Then with some calculation we can first verify $\mathbb{E}\!\left\{\bar{e}_le_l^\tr\right\}=\bar{P}$ and
\vspace{-5pt}
then  
\[
\begin{aligned}
    \mathbb{E}\left\{r_{k}r_{l}^\tr\right\}=&C\mathbb{E}\left\{e_ke_l^\tr\right\}C^\tr+C\mathbb{E}\left\{e_kv_{l}^\tr\right\}\\
    =&C\Phi_d\bar{A}^{s}\bar{P}C^\tr-C\Phi_d\bar{A}^{s-1}\bar{L}\tilde{R}
    =0.
\end{aligned}
\]

\vspace{-15pt}
{\bf Step 3:} For the case with $0<k-l<d$, we denote $s=d+l-k$ and the following auxiliary variable 
\[
\tilde{x}_{k|s}\!=\!\prod_{i=1}^{s}(A-L^{(i)}C)\bar{x}_k+\sum_{i=1}^{s}(\prod_{l=i+1}^{s}\!(A-L^{(i)}C))L^{(i)}y_{k+i-1},
\]
which indicates the optimal predictor $s$-step forward from time step $k$. We can verify that $\tilde{x}_{k-d|d}=\hat{x}_k$. Then denote
$e_{k|s}=x_{k+s}-\tilde{x}_{k|s}$ and $s=d-(k-l)$, we have
\vspace{-5pt}
\[
e_k\!=\!\Phi_{d-s}e_{k-d|s}\!+\!\sum_{i=s}^{d-1}\!\Phi_{d-i-1}w_{k-d+i}-\!\!\sum_{i=s}^{d-1}\!\Phi_{d-i-1}\tilde{v}_{k-d+i},
\vspace{-5pt}
\]
with $\tilde{v}_{k-d+i}=L^{(i+1)}v_{k-d+i}$.
We can verify that 
\vspace{-5pt}
\[
\begin{aligned}
    \mathbb{E}\left\{r_{k}r_{l}^\tr\right\}=&C\mathbb{E}\left\{e_ke_l^\tr\right\}C^\tr+C\mathbb{E}\left\{e_kv_{l}^\tr\right\}\\
    =&C\Phi_{d-s}\mathbb{E}\left\{e_{k-d|s}e_l^\tr\right\}C^\tr-C\Phi_{d-s-1}L^{(s+1)}R\\
    =&0.
\end{aligned}
\vspace{-5pt}
\]
For the last equality, it suffices to verify that the following two statements hold. The first statement is
$\mathbb{E}\left\{e_{k-d|s}e_l^\tr\right\}=P^{(s+1)},$ and the second statement is
\[(A-L^{(s+1)}C)P^{(s+1)}C^\tr-L^{(s+1)}R=0.\] 
The above two facts can be directly verified with a similar technique proposed in {\bf Step 2} together with the following result
\[P^{(i+1)}\!=\!(A\!-\!L^{(i)}C)P^{(i)}(A\!-\!L^{(j)}C)^\tr\!+\!Q\!+\!L^{(i)}RL^{(j)\tr}\]
holds for all $1\leq i\leq j\leq d$, which can also be directly verified with a similar procedure in {\bf step 2}.

\textbf{Proof of Part 2:} 
For this part, we have that the recursion \eqref{eq: steadyRecursion} is equivalent to Riccati recursion with parameter $(A, C, Q, R)$. With the result in \cite[Section 4]{Andersonoptimal}, the term $P^{(d+1)}$ will exponentially converge to the solution $P$ of the following Riccati equation $P=\Ric(A,C,Q,R,P)$.
Moreover, if we denote the $A_{cl}=A-LC$, where $L=APC\left(CPC^\tr+R\right)^{-1}$, then the matrix $A_{cl}$ is Schur stable, and
$\left\|P^{(d+1)}-P\right\|\leq M\rho(A_{cl})^d$,
where $M$ is a constant related to system parameters.
This property illustrates that the steady-state error covariance of the optimal predictor $\hat{y}_{k}$ will exponentially decay from $CP^{(d+1)}C^\tr+R$ to $CPC^\tr+R$, which quantitatively describes the information loss due to time delay. 

\textbf{Proof of Part 3: } 
We first show that the matrix sequence $P^{(1)},\cdots,P^{(d+1)}$ is monotonically increasing with $d$.
With \cite[Theorem 1]{ren2023effects}, we have $\bar{P}\leq P$, where $\bar{P}$ and $P$ are solutions to the following algebraic Riccati equations
$\bar{P}=\Ric(A,\bar{C},Q,\bar{R},\bar{P})$,
and
$P=\Ric(A,C,Q,R,P)$.
Note that with Woodbury Equality, we can rewrite the first Riccati equation as
\vspace{-5pt}
\[
\bar{P}=A\left(\bar{P}^{-1}+\bar{C}^\tr\bar{R}^{-1}\bar{C}\right)^{-1}A^\tr+Q,
\vspace{-5pt}\]
and the recursion of $P^{(l)}$ as
\vspace{-5pt}
\[
P^{(l+1)}=A\left(\big(P^{(l)}\big)^{-1}+C^\tr R^{-1}C\right)^{-1}A^\tr+Q.
\vspace{-5pt}\]
With $\bar{C}^\tr\bar{R}^{-1}\bar{C}=C^\tr R^{-1}C+C^{e\tr}\big(R^{e}\big)^{-1}C^e\ge C^\tr R^{-1} C$, we have
$P^{(2)}\ge P^{(1)}$. With mathematical induction, we further have
\[
\begin{aligned}
    P^{(l+1)}=&A\left(\big(P^{(l)}\big)^{-1}+C^\tr R^{-1}C\right)^{-1}A^\tr+Q\\
    \ge& A\left(\big(P^{(l-1)}\big)^{-1}+C^\tr R^{-1}C\right)^{-1}A^\tr+Q=P^{(l)}.
\end{aligned}
\]
Then we rewrite the recursion \eqref{eq: steadyRecursion} as
\[
P^{(i+1)}=\big(A-L^{(i)} C\big) P^{(i)}\big(A-L^{(i)} C\big)^{\tr}+Q+L^{(i)} R L^{(i) \tr},
\]
and we can obtain
\[
\big(A-L^{(d)} C\big) P^{(d)}\big(A-L^{(d)} C\big)^{\tr}\!\!\leq\! P^{^{(d+1)}}\!\!-Q\!\leq\! (1-\rho_0)P^{(d+1)}.
\]
With a similar procedure, we can also obtain 
\[
\big(A-L^{(l)} C\big) P^{(l)}\big(A-L^{(l)} C\big)^{\tr}\!\!\leq\! (1-\rho_0)P^{(l+1)},\;\forall 1\leq l\leq d.
\]
Therefore, we have the following inequality for state transition matrices
\[
\Phi_d P^{(1)}\Phi_d^\tr\leq (1-\rho_0)\Phi_{d-1}P^{(2)}\Phi_{d-1}^\tr\leq (1-\rho_0)^d P^{(d+1)}.
\]   
\end{proof}

\vspace{-5pt}
\begin{remark}
    The main difficulty of the proof lies in the asynchronous recursion of innovation $r_k$ induced by the time delay, i.e., the statement \eqref{eq: crossClaim}. Unlike the centralized Kalman filter, the innovation $r_k$ in our setup cannot be directly formulated in a recursive form in terms of its previous $r_{k-1}$. Our proof relies on a high-level insight into $\mathbb{E}\left\{r_k\tilde{r}_k^\tr\right\}=\mathbb{E}\left\{r_kr_k^\tr\right\}$, where $\tilde{r}_k$ is a suboptimal innovation corresponding to any other suboptimal predictor. We note that the orthogonality of the innovation sequence $r_{k}$ is critical to guarantee the logarithmic regret of our next model-free online prediction algorithm against the model-based optimal predictor.
\end{remark}
\vspace{-5pt}

\vspace{-1mm} 
\section{Model-free Cooperative Filtering  with Asynchronous Observations}
\label{section:algorithm&regret}

\vspace{-1mm}

In this section, we develop a least-squares-based model-free cooperative filtering algorithm. We further establish a logarithmic regret for our model-free algorithm. 

\vspace{-1mm}
\subsection{Model-free algorithm design} 
\vspace{-1mm}

It is clear that the auto-regressive model \eqref{eq: regressionModel} links the local future observations $y_{k+1}$ and partially delayed past observations $Y_{k}\bigcup Y_{k-d}^{\mathrm{e}}$. 
As guaranteed in Theorem~\ref{thm: property},  both $\Phi_d$ and $(A-\bar{L}\bar{C})^p$ have an exponential decay property. This exponential decay allows us to use a relatively short backward horizon $p$ to mitigate the bias error term induced by $\bar{x}_{k-d-p}$ in \eqref{eq: regressionModel} that may grow polynomially fast for marginally stable systems. 
Similar to \cite{tsiamis9894660, hazan2017learning,qian2025model}, we learn a least-squares estimate $\tilde{G}_{k,p}$ by regressing~$y_t$ to past outputs $Z_{t,p+d}, \;t \leq k$. 
The corresponding least-squares problem is 
\vspace{-5pt}
\begin{equation} \label{eq: leastsquare}
     \min_{G \in \mathbb{R}^{m \times (\bar{m}p+md)}} \,\,  \sum_{t=p+d}^k \left\|y_{t}-G{Z}_{t,p+d}\right\|_F^2+\lambda\left\| G\right\|_F^2, 
     \vspace{-5pt}
\end{equation}
where the regularization term $\lambda\left\| G\right\|_F^2$ is added with  a regularization parameter $\lambda >0$. At each time step $k$, by solving the ridge regression \eqref{eq: leastsquare}, we can obtain a closed-form of $G$ as 
\vspace{-5pt}
\begin{equation} \label{eq:regression-update}
   \tilde{G}_{k,p+d} = \sum_{t = p+d}^k y_t Z_{t,p+d}^\tr V_{k,p+d}^{-1}, 
   \vspace{-5pt}
\end{equation}
where  
  $  V_{k,p+d}\triangleq\lambda I + \sum_{t=p+d}^k Z_{t,p+d}Z_{t,p+d}^\tr$ 
is called a Gram matrix, which contains the information from all past observations $Z_{k,p+d}$.
 We then predict the next observation by  
 \vspace{-5pt}
\begin{equation} \label{eq:OPF-prediction}
    \tilde{y}_{k+1} = \tilde{G}_{k,p+d}Z_{{k+1},p+d}.
    \vspace{-5pt}
\end{equation}
For nonexplosive systems $\rho(A) \leq 1$, the term $\bar{x}_{k-d-p}$ in \eqref{eq: regressionModel} retains the state from previous time steps, potentially growing at a polynomial rate. A persistent bias error could result in linear regret. Fortunately, Theorem \ref{thm: property} and classical Kalman filtering theory provide two key insights: 1) the innovation $r_{k}$ is zero mean and mutually uncorrelated, and 2) $\Phi_d(A-\bar{L}\bar{C})^{p}\leq c_1\rho_0^d\rho(A-\bar{L}\bar{C})^p$, with $\rho_0<1$ and $\rho(A - \bar{L}\bar{C}) < 1$. The second property effectively controls the accumulation of bias errors $\Phi_d(A-\bar{L}\bar{C})^{p}\bar{x}_{k-p-d}$ effectively. Specifically, we can gradually increase the past horizon $p$, following the strategy in \cite{tsiamis9894660}, motivated by the standard ``doubling trick'' \cite{cesa2006prediction}. In particular, the entire time horizon is partitioned into multiple epochs, with each twice as long as the previous one. Within each epoch, the value of backward horizon length $p$ remains constant. As $\rho(A - \bar{L}\bar{C})^p$ decreases exponentially with $p$, it suffices to increase $p$ slowly as $p = \mathcal{O}(\log T)$, where $T$ is the length of each epoch.

\begin{algorithm}[t]
   \caption{Online Cooperative Filtering with Asynchronous Observations (\texttt{co-Filter})}\label{algPrediction}
\begin{algorithmic}

   \STATE {\bfseries Input:} parameter $\beta, \lambda, T_{\text {init }}, N_E$ 
   
   \COMMENT{\textit{Warm Up}:}
   \FOR{$k=0$
   {\bfseries to }$T_{\text{init}}$}  
   \STATE Observe $y_{k}$, Receive $y_{k-d}^{\mathrm{e}}$;
   \ENDFOR
   
   \COMMENT{\textit{Recursive Online Prediction:}}
    \FOR{$l=1$ {\bfseries to} $N_E$}
    \STATE
    Initialize 
      $T_l\!=\!2^{l-1} T_{\text {init }}+1, p\!=\!\beta \log T_l,$ 
    \STATE
    Compute $V_{T_l-1,p+d}$ and  $\tilde{G}_{T_l-1,p+d}$; 
    \FOR{$k=T_l$ \textbf{to} $2 T_l-2$}
    \STATE Predict $\tilde{y}_{k}=\tilde{G}_{k-1,p+d} Z_{k,p+d}$; 
    \STATE
    Observe $y_{k}$, Receive $y_{k-d}^{\mathrm{e}}$;
    \STATE
    Update $V_{k,p+d}$ and $\tilde{G}_{k,p+d}$ as \eqref{eq:predictor-update}. 
    \ENDFOR
    \ENDFOR   
\end{algorithmic}
\end{algorithm}

The pseudo-code for online cooperative filtering with asynchronous observations (\texttt{co-Filter}) is presented in Algorithm~\ref{algPrediction}. It has two main phases: \textit{Warm-Up} and \textit{Online Prediction}. In the warm-up phase, we collect a trajectory of observations of length $T_{\mathrm{init}}$, i.e., $Y_{T_{\text{init}}}\bigcup Y_{T_{\text{init}}-d}^{\mathrm{e}}$. In the phase of online prediction, we first initialize the parameter $p$, $V_{T_l-1,p+d}$, and $G_{T_l-1,p+d}$ at the beginning of each epoch $T_l$, where $l$ is the index of the epoch number.  At each time step $k$, we update the prediction $\tilde{y}_k$  using \eqref{eq:OPF-prediction} and then observe the new observation ${y}_k, y_{k-d}^{\mathrm{e}}$ simultaneously. Then the predictor $G_{k,p+d}$ can be recursively updated with
\begin{minipage}{0.48\textwidth}
\begin{align}
V_{k,p+d}&=V_{k-1,p+d}+Z_{k, p+d} Z_{k, p+d}^{\tr}, \label{eq:predictor-update} \\
\tilde{G}_{k,p+d}&=\tilde{G}_{k-1,p+d}+\left(y_{k}-\tilde{y}_{k}\right) Z_{k, p+d}^{\tr}V_{k,p+d}^{-1}. \nonumber 
\end{align}
\end{minipage}

Algorithm~\ref{algPrediction} is model-free and requires no knowledge of the system model or noise covariances.

\vspace{-1.5mm}
\subsection{Logarithmic regret against the optimal cooperative predictor}

\vspace{-1mm}

To measure the performance of \texttt{co-Filter} in Algorithm~\ref{algPrediction}, we consider the regret below: 
\begin{equation}\label{regret-new}
    \mathcal{R}_{N} \triangleq \sum_{k=1}^{N}\left\|y_{k}-\tilde{y}_{k}\right\|^{2}-\sum_{k=1}^{N}\left\|y_{k}-\hat{y}_{k}\right\|^{2},
\end{equation}
where $\hat{y}_{k}$ is from the optimal model-based cooperative prediction in Theorem~\ref{theorem:optimal-delayed-filter}. This is similar to \eqref{regret} used in \cite{tsiamis9894660, hazan2017learning,qian2025model}, but with one key difference: the benchmark in \eqref{regret-new} is from  Theorem~\ref{theorem:optimal-delayed-filter}, which is much stronger than that in \eqref{regret}. We have a logarithmic regret guarantee.   

\vspace{6pt}
\begin{theorem}\label{thm: regret}
    Consider the linear stochastic system \eqref{eq: LinearSystem} with an external information source \eqref{eq:external-source} with i.i.d. Gaussian noises.  Suppose Assumption~\ref{asp: Ob-main-text} and \ref{asp: Diagonal} hold. Choose the parameters in Algorithm~\ref{algPrediction}  as  \[\beta=\frac{\Omega(\kappa)}{\log (1 / \rho(A-\bar{L} \bar{C}))},\;\; 
    T_{\mathrm{init}} = \operatorname{poly}\left(d,\beta,\log\left(\frac{1}{\delta}\right)\right), \] where $\kappa$ represents the order of the largest Jordan block of $A$ corresponding to the eigenvalue 
$1$. Fix a horizon $N > T_{\mathrm{init}}$ and a failure probability $\delta > 0$. Then  with probability at least $1-\delta$, we have 
    $$
    \mathcal{R}_{N} \leq \operatorname{poly}\left(d,\beta,\log(1/\delta)\right) \mathcal{O}\left(\log^{3} N\right),$$
    where $M$ is a constant only related to the system parameter.
\end{theorem}

   Intuitively, achieving logarithmic regret is expected, since we learn the regression model \eqref{eq: regressionModel} within the online least-squares framework \cite{draper1998applied}.     
    Given the fact that $(A-\bar{L}\bar{C})^p$ is exponentially decaying and the innovation noise $r_k$ is mutually uncorrelated by Theorem~\ref{thm: property}, selecting a large $p=\mathcal{O}(\log T_l)$ in each epoch is sufficient to obtain an online predictor $\tilde{y}_k$ with logarithmic regret.  
    In Theorem~\ref{thm: regret}, we have established a much shaper bound than the existing results \cite{tsiamis9894660, rashidinejad2020slip}, where similar logarithmic regret bounds have order $\mathcal{O}\left(\log^{6} N\right)$ and $\mathcal{O}\left(\log^{11} N\right)$, respectively (note that both \cite{tsiamis9894660, rashidinejad2020slip} only focused on centralized Kalman filter without asymmetric data structure).  

    The detailed proof for Theorem \ref{thm: regret} is technically involved, and we present it in the next subsection.

\vspace{-1mm}

\subsection{Proof of Theorem~\ref{thm: regret}}
\vspace{-1mm}

We outline the proof of Theorem~\ref{thm: regret} in four main steps. The main proof framework can follow standard online least-squares-based methods \cite{lai1991recursive,tsiamis9894660,rashidinejad2020slip}, and the main technical difficulty of this theorem lies in handling the asymmetry in the data structure induced by time delay. To simplify the proof process and highlight the theoretical contribution, we will present the overlap proof using a technique similar to the literature and highlight the key difficulty in the following proof. 

With \cite[Theorem 1]{tsiamis9894660}, the regret $\mathcal{R}_N$ is dominated by $\mathcal{L}_N\triangleq \sum_{k=T_{\textnormal{init}}}^{N}\left\|\tilde{y}_{k}-\hat{y}_{k}\right\|_{2}^{2}$, i.e., $\mathcal{R}_N=\mathcal{L}_N+o(\mathcal{L}_N)$. Hence, in the following proof, we mainly analyze the scaling of $\mathcal{L}_N$ with respect to $N$ for fixed $d$. We decompose each $\tilde{y}_{k+1}-\hat{y}_{k+1}$ into three parts in a standard manner,
$$
\begin{aligned}
\tilde{y}_{k+1}-\hat{y}_{k+1}=  \underbrace{-\lambda G_{p+d}  V_{k, p+d}^{-1} Z_{k+1, p+d}}_{\text{Regularization error}}\qquad\qquad\qquad
\\+\underbrace{\sum_{l=p+d}^{k} b_{l, p+d} Z_{l, p+d}^{{\tr}}V_{k, p+d}^{-1} Z_{k+1, p+d}-b_{k+1, p+d}}_{\text{Bias error}}  \\
+\underbrace{\sum_{l=p+d}^{k} r_{l} Z_{l, p+d}^{{\tr}} V_{k, p+d}^{-1} Z_{k+1, p+d}}_{\text{Regression error}},\qquad\qquad
\end{aligned}
$$
where $b_{k,p+d}=C\Phi_d(A-\bar{L} \bar{C})^{p} \bar{x}_{k-d-p}$ denotes the bias in the autoregressive model \eqref{eq: regressionModel}.
By further denoting $\mathcal{B}_{k}\triangleq \Big\| \sum_{l=p}^{k} b_{l, p} Z_{l, p+d}^{\tr} V_{k, p+d}^{-\frac{1}{2}}\Big\|_{2}^{2}$ as the bias factor, $\mathcal{E}_{k}\triangleq \Big\| \sum_{l=p}^{k} r_{l} Z_{l, p+d}^{\tr} V_{k, p+d}^{-\frac{1}{2}}\Big\|_{2}^{2}$ as the regression factor, $\mathcal{G}_{k}\triangleq \Big\| \lambda G_{p+d} V_{k, p+d}^{-\frac{1}{2}}\Big\|_{2}^{2}$ as the regularization factor, $\mathcal{V}_{N}\triangleq\!\!
    \sum_{k=T_{\textnormal{init}}}^{N}\left\|V_{k-1, p+d}^{-\frac{1}{2}} Z_{k,p+d}\right\|_{2}^{2}$
    as the accumulation factor, and $
\mathfrak{b}_{N}=\sum_{k=T_{\textnormal{init}}}^{N}\left\|b_{k, p+d}\right\|_{2}^{2}$ as the accumulation of bias. Then we can further decompose the term $\mathcal{L}_N$ into
\begin{equation}\label{eq: decoupleRegretBound}
        \mathcal{L}_{N}\leq 4\big(\max_{k\leq N
}\left(\mathcal{B}_{k} +\mathcal{E}_{k} + \mathcal{G}_{k}\right)\big)\cdot\mathcal{V}_{N}+4\mathfrak{b}_{N}.
 \end{equation}
Here, the subscript $p$ in the above factors is eliminated because $p$ is uniquely determined by $k$. Note that from the expression of $G_{p+d}$, we can obtain that 
   $
    \mathcal{G}_{k}\leq \left\|G_{p+d}G_{p+d}^{\tr}\right\|_2^2 \leq \frac{M}{1-\rho(A-\bar{L}\bar{C})^2}, \,\forall k\ge T_{\text{init}}
    $ does not scale with $N$, where the last inequality is from the uniform boundedness of $\Phi_l,\;l=1\dots,d$ in Theorem \ref{thm: property}, together with the diagonalizability of $A-\bar{L}\bar{C}$ in Assumption \ref{asp: Diagonal}. 
    
    Thus, in the following proof, we only need to discuss the bias factor $\mathcal{B}_{k}$, the regression factor $\mathcal{E}_{k}$, and the accumulation factor $\mathcal{V}_{N}$. Then the term-wise analysis for the above factors is presented below.
    
\vspace{2mm}
 
    \begin{lemma} \label{lemma:bias-factor}
Suppose Assumption~\ref{asp: Ob-main-text} and \ref{asp: Diagonal} hold, then for any fixed $\beta$ satisfying the condition in Theorem~\ref{thm: regret}, we have 
$$\max_{T_{\textnormal{init}}\!\leq k\! \leq N}\mathcal{B}_{k}\!\leq \!M\poly(\beta,\log(1/\delta))\!\log N
$$
and $\;
\mathfrak{b}_{N}\!\leq\! M\poly(\beta,\log(1/\delta))\!\log^2N\;$
holds for all $N$ with probability at least $1-\delta$.
\end{lemma}

\vspace{5pt}

\begin{lemma} \label{lemma:regression-factor}
Suppose Assumption~\ref{asp: Ob-main-text} holds. For any fixed $\beta$ satisfying the condition in Theorem~\ref{thm: regret}, we have 
\begin{equation} 
\max_{T_{\textnormal{init}}\leq k \leq N}\mathcal{E}_{k}\!\leq \!M\poly(\beta,\log(1/\delta))\log N
\end{equation}
holds for all $N$ with probability at least $1-\delta$.
\end{lemma}

\vspace{5pt}

\begin{lemma} \label{lemma:accumulation-factor}
Suppose Assumption~\ref{asp: Ob-main-text} holds. For any fixed $\beta$ satisfying the condition in Theorem~\ref{thm: regret}, we have 
\begin{equation} 
\mathcal{V}_{N}\!\leq \!M\poly(\beta,\log(1/\delta))\log^2 N
\end{equation}
holds for all $N$ with probability at least $1-\delta$.
\end{lemma}

{The proof ideas of Lemma~\ref{lemma:bias-factor} and \ref{lemma:regression-factor} generally follow the procedure in \cite[Theorem 1]{tsiamis9894660} and \cite[Theorem 1]{qian2025model}, with some additional derivations to handle the effect induced by asynchronous observations. {The detailed proofs are provided in the Appendix \ref{sectionBeta} and \ref{sectionRegression}.} }
{While the proof of Lemma~\ref{lemma:accumulation-factor} contains several non-trivial mathematical techniques, such as the persistent excitation analysis for the asymmetry Gram matrix and corresponding asymmetric successive reformulation to guarantee a small regret order. To highlight the core theoretical contribution, we present the essential proof steps of Lemma~\ref{lemma:accumulation-factor} in the Appendix~\ref{sec:proofAccumulation}.} 

\textbf{Proof of Theorem~\ref{thm: regret}}: We can now directly combine \eqref{eq: decoupleRegretBound} and Lemma~\ref{lemma:bias-factor} to \ref{lemma:accumulation-factor} to establish the regret bound as
\[
\begin{aligned}
    \mathcal{L}_{N}\leq& 4\big(\max_{k\leq N
}\left(\mathcal{B}_{k} +\mathcal{E}_{k} + \mathcal{G}_{k}\right)\big)\cdot\mathcal{V}_{N}+4\mathfrak{b}_{N}\\
\leq & M\poly(\beta,\log(1/\delta))\Big(\log N+\log N+c\Big)\\
&\times\log^2 N+M\poly(\beta,\log(1/\delta))\log^2 N\\
=&M\poly(\beta,\log(1/\delta))\log^3 N,
\end{aligned}
\]
where $c=\frac{M}{1-\rho(A-\bar{L}\bar{C})^2}$ denotes the uniform upperbound for the regularization factor $\mathcal{G}_k$ and we eliminate the low-order terms.

\section{Performance Improvement of Online Cooperative Predictor over Local Optimal Predictor}
In the previous section, we have provided a regret analysis of the online partially delayed cooperative filter relative to its optimal model-based counterpart. In this section, we establish the condition to guarantee that, for sufficiently large $N$, the cooperative online filter outperforms the optimal model-based local predictor. 
Accordingly, the performance metric will be replaced by
\begin{equation}\label{regret-local}
    \tilde{\mathcal{R}}_{N} \triangleq \sum_{k=1}^{N}\left\|y_{k}-\tilde{y}_{k}\right\|^{2}-\sum_{k=1}^{N}\left\|y_{k}-\check{y}_{k}\right\|^{2},
\end{equation}
where the term $\check{y}_{k}$ denotes the optimal prediction with only local information $Y_{0:k-1}$, i.e., 
\begin{equation}\label{MMSEProblem-local}
    \check{y}_{k+1} \triangleq \arg \min _{z \in \mathcal{F}_{k}} \mathbb{E}\left[\left\|y_{k+1}-z\right\|_{2}^{2} \mid \mathcal{F}_{k}\right], 
\end{equation}
with $\mathcal{F}_{k}\triangleq \left\{y_0,\dots,y_{k-1}\right\}$.  

It is clear that we have  
\begin{equation}\label{eq: modifyregretDecompose}
    \tilde{\mathcal{R}}_{N}=\mathcal{R}_N+\sum_{k=1}^{N}\left\|y_{k}-\hat{y}_{k}\right\|^{2}-\sum_{k=1}^{N}\left\|y_{k}-\check{y}_{k}\right\|^{2},
\end{equation}
which shows that the modified regret comprises two components: the learning regret $\mathcal{R}_N$ and the model-based regret. The model-based regret captures the fundamental performance improvement enabled by incorporating additional observations, whereas the learning regret arises from the lack of model information. Consequently, for the online learning method to outperform the local model-based predictor, the performance gain from additional information should dominate the learning regret.

Before presenting our results on performance improvement, we first provide an example to illustrate that incorporating multiple-source data does not necessarily enhance local prediction performance; in particular, $\sum_{k=1}^{N}\left\|y_{k}-\hat{y}_{k}\right\|^{2}-\sum_{k=1}^{N}\left\|y_{k}-\check{y}_{k}\right\|^{2}$ can be 0. 
\vspace{6pt}
\begin{example} \label{example: fakeimprove}
    Consider a second-order system with the system model and noise statistics below
    \begin{equation}\label{eq: exampleParaSup}
    \begin{aligned}
        A&=\begin{bmatrix}
    0.9&0\\
    0&0.9
\end{bmatrix}, \quad \bar{C}=\begin{bmatrix}
    C\\C^{\mathrm{e}}
\end{bmatrix}=
\begin{bmatrix}
1&0\\0&1
\end{bmatrix},\\  
Q &= \begin{bmatrix}
    1&0\\0&1
\end{bmatrix},
\qquad \bar{R} =\begin{bmatrix}
R\\&R^{\mathrm{e}}
\end{bmatrix}=\begin{bmatrix}
    1&0\\0&1
\end{bmatrix}.
    \end{aligned} 
\end{equation}
By solving the corresponding Riccati equations, we can obtain
the local covariance and the centralized covariance as
\[
P=\begin{bmatrix}
    1.48&0\\
    0&5.26
\end{bmatrix},\quad
\bar{P}=\begin{bmatrix}
    1.48&0\\
    0&1.48
\end{bmatrix},
\]
which indicates that $CPC^\tr-C\bar{P}C^\tr=0$, and that the optimal cooperative predictor cannot improve local prediction performance.
  \hfill $\square$
\end{example}
The above example shows that, for unrelated observations, collecting more information cannot improve prediction performance. To provide a theoretical guarantee, we impose the following condition on the inherent observation property to ensure prediction improvement.
We define the following {\it symplectic matrices}
\[
\bar{S}\!=\!\begin{bmatrix}
    A^{T}\!\!+\!\bar{G} \!A^{-1} \!Q & \;-\bar{G} A^{-1} \!\\
-A^{-1} Q & A^{-1}
\end{bmatrix}\!\!,{S}\!=\!\begin{bmatrix}
    A^{T}\!\!+\!G\! A^{-1}\!Q & \;-G A^{-1}\! \\
-A^{-1} Q & A^{-1}
\end{bmatrix}
\]
where $G=C^\tr R^{-1}C$ and $\bar{G}=\bar{C}^\tr \bar{R}^{-1}\bar{C}=C^\tr R^{-1}C+C^{(e)\tr} R^{(e)-1}C^{(e)}$.
Then we provide the following Assumption to guarantee strict performance improvement.
\vspace{1pt}

\begin{assumption}(\cite[Theorem 2]{ren2023effects})\label{asp: mono}
    System matrix $A$ is invertible, and the corresponding symplectic matrices $\bar{S}$ and $S$ do not have any common
stable eigenvalues and associated common eigenvectors
simultaneously.
\end{assumption}
\vspace{-5pt}

With \cite[Theorem 2]{ren2023effects},  Assumption~\ref{asp: mono} guarantees the strict inequality $P>\bar{P}$, implying that additional information yields a strict performance improvement in the centralized case. However, centralized performance improvement does not necessarily imply performance improvement under partial information delay. We have the following theorem to characterize performance improvement in the delayed setting. 

\vspace{6pt}
\begin{theorem}\label{thm: improvement}
   Consider the linear stochastic system \eqref{eq: LinearSystem} with an external information source \eqref{eq:external-source} with i.i.d. Gaussian noises.  Suppose Assumption~\ref{asp: Ob-main-text} and \ref{asp: mono} hold. Fix a horizon $N > T_{\mathrm{init}}$ and a failure probability $\delta > 0$. Then  with probability at least $1-\delta$, we have 
   \vspace{-5pt}
    $$
    \sum_{k=1}^{N}\!\left\|y_{k}\!-\!\hat{y}_{k}\right\|^{2}-\sum_{k=1}^{N}\!\left\|y_{k}\!-\!\check{y}_{k}\right\|^{2}\!\leq\! -\mathcal{O}(N)+\poly\!\big(\frac{1}{\delta}\big)\mathcal{O}(\sqrt{N})
    \vspace{-5pt}$$
    holds for all $N$ uniformly.
\end{theorem}
\begin{proof}
With Assumption~\ref{asp: mono}, we first have that
$
P^{(l+1)}=A\big(P^{(l)-1}+C^\tr R^{-1}C\big)^{-1}A^\tr+Q.
$
Then we further have
\vspace{-5pt}
\[
\begin{aligned}
  P^{(l+1)}-P=&A\big(P^{(l)-1}+C^\tr R^{-1}C\big)^{-1}A^\tr \\
  &\qquad\qquad-A\big(P^{-1}+C^\tr R^{-1}C\big)^{-1}A^\tr.
\end{aligned}
\vspace{-5pt}
\]
With $A$ invertible and mathematical induction, if $P^{(l)}<P$, then we have $P^{(l+1)}<P$ also holds. Therefore with Assumption~\ref{asp: mono}, we have $\bar{P}=P^{(1)}<P$. Hence, we have $P^{(d)}<P$ holds for all $d$, which means that the delayed additional information will result in strict performance improvement.

Then we denote $\check{r}_k=y_{k}-\check{y}_{k}$, and with $r_k=y_{k}-\hat{y}_{k}$ in Theorem~\ref{thm: property}, we have 
$
\mathbb{E}\left\{r_kr_k^\tr\right\}=CP^{(d+1)}C^\tr+R\triangleq \hat{R},\, \text{and } \mathbb{E}\left\{\check{r}_k\check{r}_k^\tr\right\}=CPC^\tr+R\triangleq \check{R}.
$
Then we have $\sum_{k=1}^{N}\left\|y_{k}-\hat{y}_{k}\right\|^{2}=\sum_{k=1}^{N} r_k^\tr r_k$ and $\sum_{k=1}^{N}\left\|y_{k}-\check{y}_{k}\right\|^{2}=\sum_{k=1}^{N} \check{r}_k^\tr \check{r}_k$. Consider the random variable $\sum_{k=1}^{N} r_k^\tr r_k$, we first have
$\mathbb{E}\left\{\sum_{k=1}^{N} r_k^\tr r_k\right\}=N\trace(\hat{R})$.
Then consider the quadratic reformulation of $\sum_{k=1}^{N} r_k^\tr r_k$ as
\vspace{-5pt}
\[
\sum_{k=1}^{N} r_k^\tr r_k=r_{1:N}^\tr \big(I_N\otimes \hat{R}^{-\frac{1}{2}}\big)\big(I_N\otimes\hat{R}\big)\big(I_N\otimes \hat{R}^{-\frac{1}{2}}\big)r_{1:N},
\vspace{-5pt}\]
where $r_{1:N}^\tr\triangleq\left[ r_1^\tr,\dots,r_N^\tr\right]$ is the collection of innovations. It is easy to verify that the vector $\big(I_N\otimes \hat{R}^{-\frac{1}{2}}\big)r_{1:N}$ is composed by standard Gaussian random variable. Then for the matrix $I_N\otimes \hat{R}$, we can also verify that  $\big\|I_N\otimes \hat{R}\big\|_F^2=N\big\|\hat{R}\big\|_F^2$ and $\big\|I_N\otimes \hat{R}\big\|=\big\|\hat{R}\big\|$. With Lemma~\ref{lm: HWinequality} (Hanson-Wright inequality), for fixed probability $\delta$, we can choose 
\[
t_N=\max\left\{c\sqrt{N}\big\|\hat{R}\big\|_F\sqrt{\log\frac{N^2}{\delta}},\;\;c\big\|\hat{R}\big\|\log\frac{N^2}{\delta}\right\},
\]
where $c$ is the constant that is only related to the standard Gaussian random variable. 
Then we have
\[
\mathbb{P}\left\{\Big|\sum_{k=1}^{N} r_k^\tr r_k-N\trace(\hat{R})\Big|>t_N\right\}\leq \frac{\delta}{N^2}
\]
holds for each $N$. By taking a union for all $N$, we can approximately obtain the result that
\[
\mathbb{P}\left\{\!\Big|\sum_{k=1}^{N} r_k^\tr r_k\!-\!N\trace(\hat{R})\Big|\!<\!c\sqrt{N}\big\|\hat{R}\big\|_F{\log\frac{N^2}{\delta}},\forall N\!\right\}\!\ge\! 1-\delta
\]
holds, where $c\sqrt{N}\big\|\hat{R}\big\|_F{\log\frac{N^2}{\delta}}$ is an approximate relaxation of $t_N$.
With a similar process, we can verify that
\[
\mathbb{P}\left\{\!\Big|\sum_{k=1}^{N} \check{r}_k^\tr \check{r}_k\!-\!N\trace(\check{R})\Big|\!<\!c\sqrt{N}\big\|\check{R}\big\|_F{\log\frac{N^2}{\delta}},\forall N\!\right\}\!\ge\! 1-\delta
\]
also holds. Further relaxing the above two inequalities, we can obtain that
\vspace{-5pt}
\[
\begin{aligned}
    \sum_{k=1}^{N} r_k^\tr r_k-\sum_{k=1}^{N} \check{r}_k^\tr \check{r}_k\leq& N\trace(\hat{R})-N\trace(\check{R})\\
    &+c\sqrt{N}\big(\big\|\check{R}\big\|_F+\big\|\hat{R}\big\|_F\big){\log\frac{N^2}{\delta}}.
\end{aligned}
\vspace{-5pt}
\]
holds uniformly for all $N$ with probability $1-2\delta$.
\end{proof}

Together with \eqref{eq: modifyregretDecompose}. Note that $\trace(\hat{R})<\trace(\check{R})$, then for large enough $N$, the $\sqrt{N}$ and $\log N$ terms in the regret $\tilde{\mathcal{R}}_N$ will be dominated by the negative linear term. 
Then we can see that the worst-case performance of the cooperative online filter will outperform the best-case performance of the optimal local filter for sufficiently large $N$, and this dominance relation holds uniformly for all $N$. Then we have the following result.

\vspace{3pt}
\begin{corollary}\label{coro: improvement}
   Consider the linear stochastic system \eqref{eq: LinearSystem} with an external information source \eqref{eq:external-source} with i.i.d. Gaussian noises.  Suppose Assumption~\ref{asp: Ob-main-text}, \ref{asp: Diagonal} and \ref{asp: mono} hold. Choose the parameters in Algorithm~\ref{algPrediction} the same as those in Theorem~\ref{thm: regret}. Fix a horizon $N > T_{\mathrm{init}}$ and a failure probability $\delta > 0$. Then  with probability at least $1-\delta$, we have 
    $$
    \tilde{\mathcal{R}}_{N} \leq - \mathcal{O}\left(N\right),$$
    holds for all $N\ge M\poly(\log 1/\delta)$
    where $M$ is a constant only related to the system parameter.
\end{corollary}

    This corollary highlights the advantage of guaranteeing logarithmic regret in online learning. When additional information is incorporated into local prediction, the resulting fundamental performance improvement scales linearly with $N$. Consequently, if the learning regret grows sublinearly, there exists a sufficiently long horizon $N$, such that the learning-based method outperforms the locally optimal model-based predictor. This result demonstrates the effectiveness of leveraging additional related data: even delayed information can still improve prediction performance. Note that the information loss induced by time delay grows exponentially with the delay length $d$; therefore, the horizon required to observe performance improvement also scales at least exponentially with $d$.

\section{Numerical Experiments}\label{sec: simulation}

We first discuss a parameter tuning method for our cooperative filter algorithm and then provide two numerical experiments to verify the logarithmic regret guarantee of Algorithm~\ref{algPrediction} (Theorem \ref{thm: regret}), and performance improvement over the local optimal predictor (Theorem~\ref{thm: improvement}). , 

\subsection{Parameter tuning in \texttt{co-Filter}}\label{subsec: paramterTuning}

In Theorem~\ref{thm: regret}, the hyperparameter $\beta$ determines the effective look-back window in each epoch and must be carefully tuned for practical implementation. While  Theorem~\ref{thm: regret} primarily provides theoretical guarantees, it is equally important to discuss how to select hyperparameters in real-world scenarios. To guide the choice of 
$\beta$ in practice, we propose an \textit{Ensemble-Based Selection Method}. Specifically,
we run $M$ parallel predictors, each using a different parameter $\beta^{(i)}$, but sharing the same observed data stream $Y_{0:k}\bigcup Y_{0:k-d}^e$. Denote the predictions by $\tilde{y}_{k+1}^{(i)}$, for $i = 1, \dots, M$. At each time step $k$, the final prediction $\tilde{y}_k$ is chosen from the predictor with the lowest accumulated prediction error:
\vspace{-1mm}
\[
\tilde{y}_{k+1} = \tilde{y}_{k+1}^{(i_0)}, \;
\text{where} \; i_0 = \arg\min_{i=1,\dots,M} \sum_{l=T_{\text{init}}}^{k} \left\| y_{l} - \tilde{y}_{l}^{(i)} \right\|^2. 
\]
As long as one parameter $\beta^{(l_0)}$ satisfies the conditions in Theorem~\ref{thm: regret}, then the prediction series $\tilde{y}_{k+1}$ achieves logarithmic regret with respect to the model-based optimal cooperative predictor.

\subsection{Experimental results}

\begin{figure}[t]
    \centering
    \subfigure[]{\includegraphics[width=0.4\textwidth]{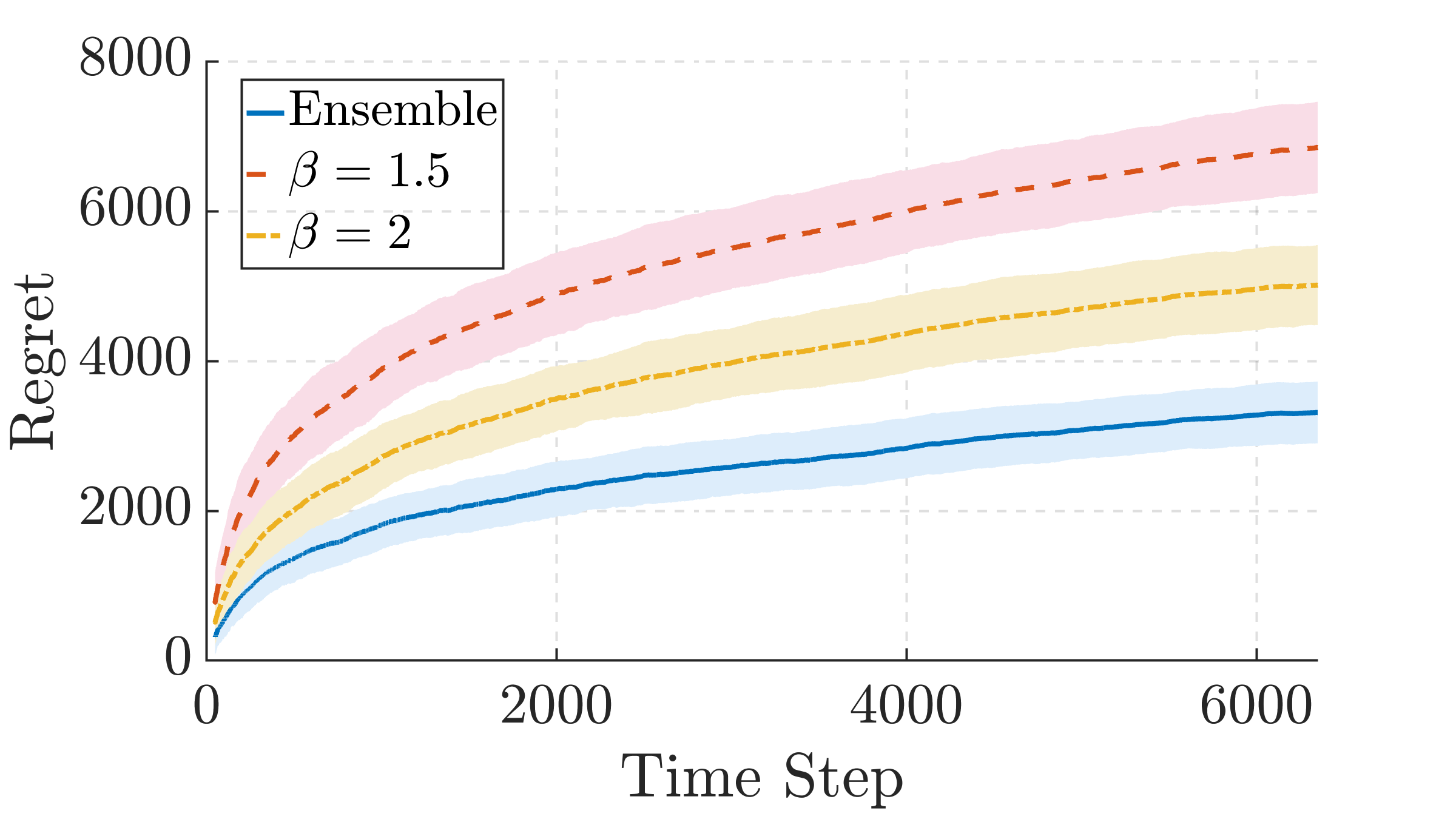}
\label{figEnsemble}}
\hspace{6mm}
     \subfigure[]
{\includegraphics[width=0.4\textwidth]{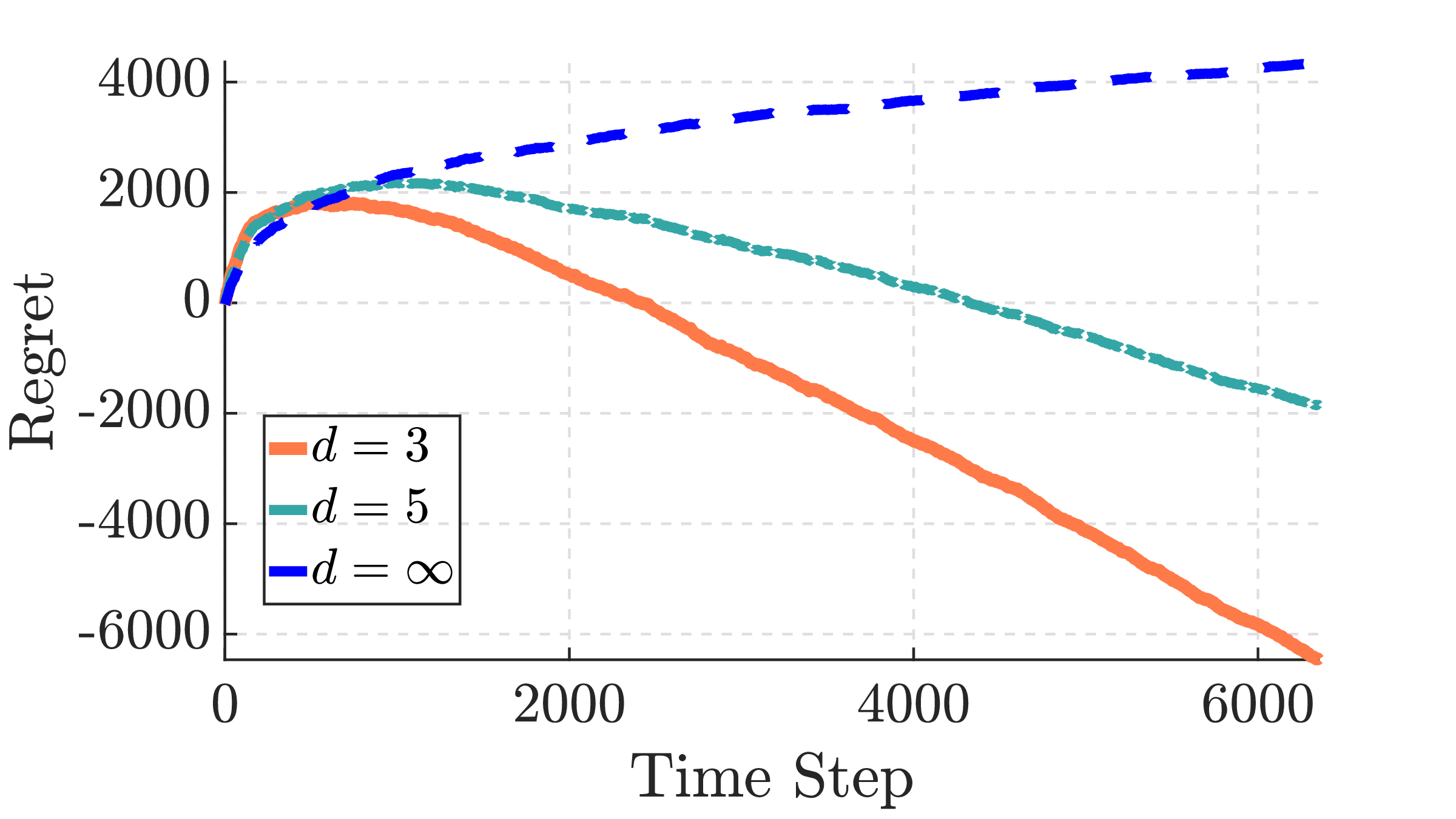}
\label{figMultidelay}}
    \caption{Logarithmic regret in Experiment 1. (a) Regret performance of the {\it Ensemble-based Method} in Section~\ref{subsec: paramterTuning} for hyperparameter tuning in  Theorem~\ref{thm: regret}; (b) Performance improvement of our \texttt{co-Filter} in Algorithm~\ref{algPrediction} with different time delays  $d=3,5,\infty$ compared to a standard Kalman filter with only local observations. When $d=\infty$, our \texttt{co-Filter} is reduced to the online algorithm in \cite{tsiamis9894660}; but with a finite time delay $d$, our \texttt{co-Filter} has a better performance than \cite{tsiamis9894660} and achieves \textit{negative} regret against the standard Kalman filter with only local observations, which is a weaker benchmark.   } 
\end{figure}

{\bf Experiment 1: a consensus-type system.} In our first experiment, the system matrix $A$ in \eqref{eq: LinearSystem} is a randomly generated $n\times n$ stochastic matrix, i.e., the summation of each row of $A$ equals 1. In this numerical experiment, we choose $n=10$. This type of system appears in many consensus-type applications, such as the formation of multi-robot swarms \cite{fax2004information}, brain networks \cite{gu2015controllability}, and vehicle platooning \cite{zheng2015stability}. The covariance matrix $Q$ for the process noise $w_k$ is also randomly generated.
The other parameters are set as $C=\boldsymbol{e}_1,\, C^{\mathrm{e}}=\boldsymbol{e}_2,\, R=R_e=0.01,$ 
where $\boldsymbol{e}_i \in\mathbb{R}^{1\times n}$ is a vector with the $i$-th element being 1 and all other elements being 0.  
The parameters in Algorithm~\ref{algPrediction} are set as
$T_{\text{init}}=50,\; \lambda=1,\; N_E = 7.$
To verify the effectiveness of the {\it Ensemble-based Selection Method}, we here simultaneously choose multiple $\beta$ for online prediction. The specific value of $\beta$ is chosen as
$\beta=\left[1,1.25,1.5,1.75,2,2.25,2.5\right].$ In the first part, the delay step $d$ is set to be $d=1.$ 
Figure~\ref{figEnsemble} illustrates the regret performance of our \texttt{co-Filter} in Algorithm~\ref{algPrediction}. As expected in Theorem~\ref{thm: regret}, the results in Figure~\ref{figEnsemble} verify the logarithmic regret against the optimal delayed predictor in Theorem~\ref{theorem:optimal-delayed-filter}, where the shaded area represents 1-$\sigma$ deviation. Furthermore, as long as there exists one $\beta^{(i)}$ such that the corresponding regret is logarithmic to $N$, then the {\it Ensemble-based Method} can guarantee the logarithmic regret of Algorithm~\ref{algPrediction}. With the designed selection strategy, the performance of the {\it Ensemble-based Method} outperforms that of other choices of $\beta$.

\begin{figure*}[t]
    \centering
    \subfigure[Real \& Predicted Trajectory.]{\includegraphics[width=0.29\textwidth]{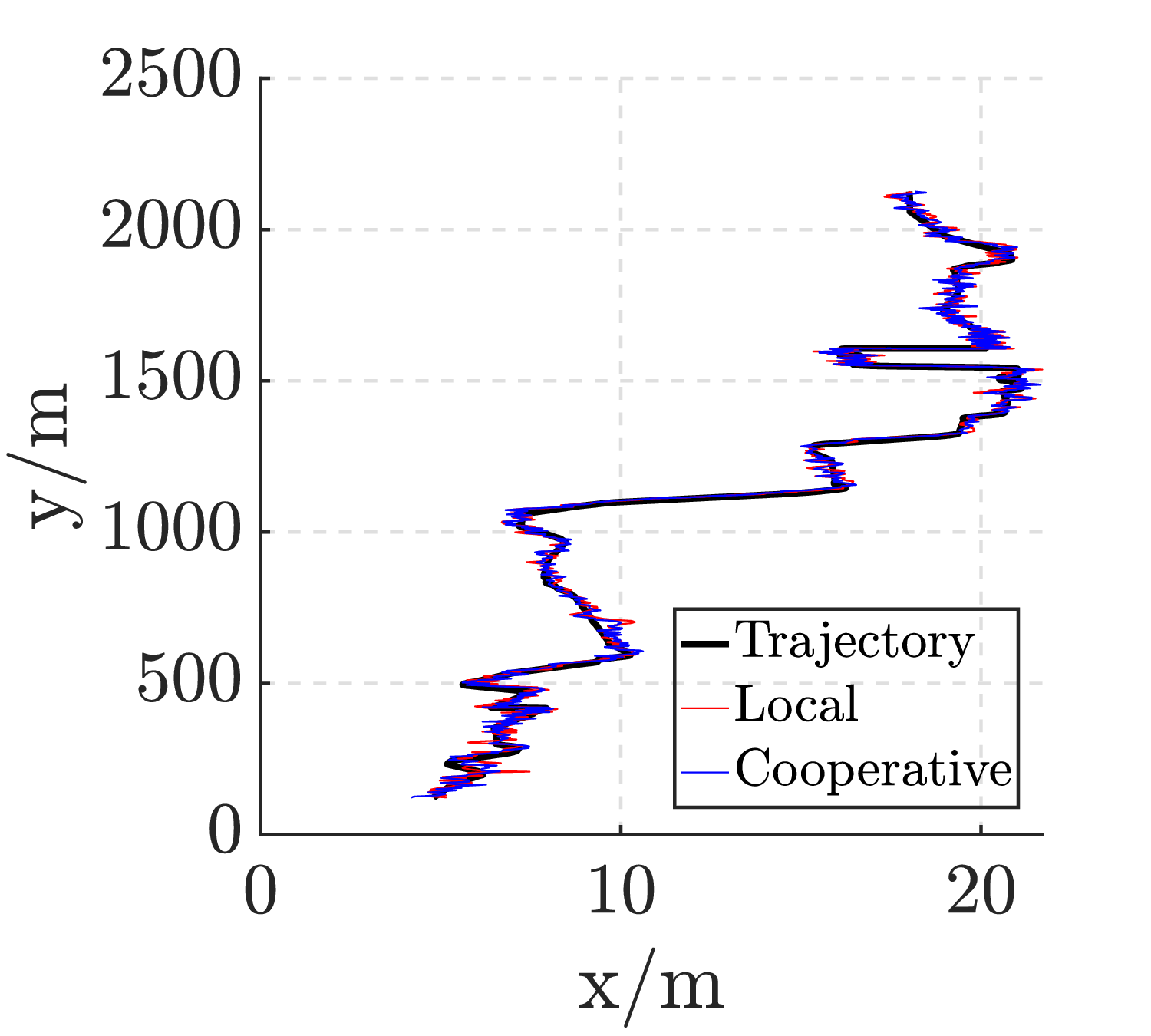}
    \label{figTrajectory}}  \hspace{1mm}
    \subfigure[Comparison of Different  Delay.]{\includegraphics[width=0.29\textwidth]{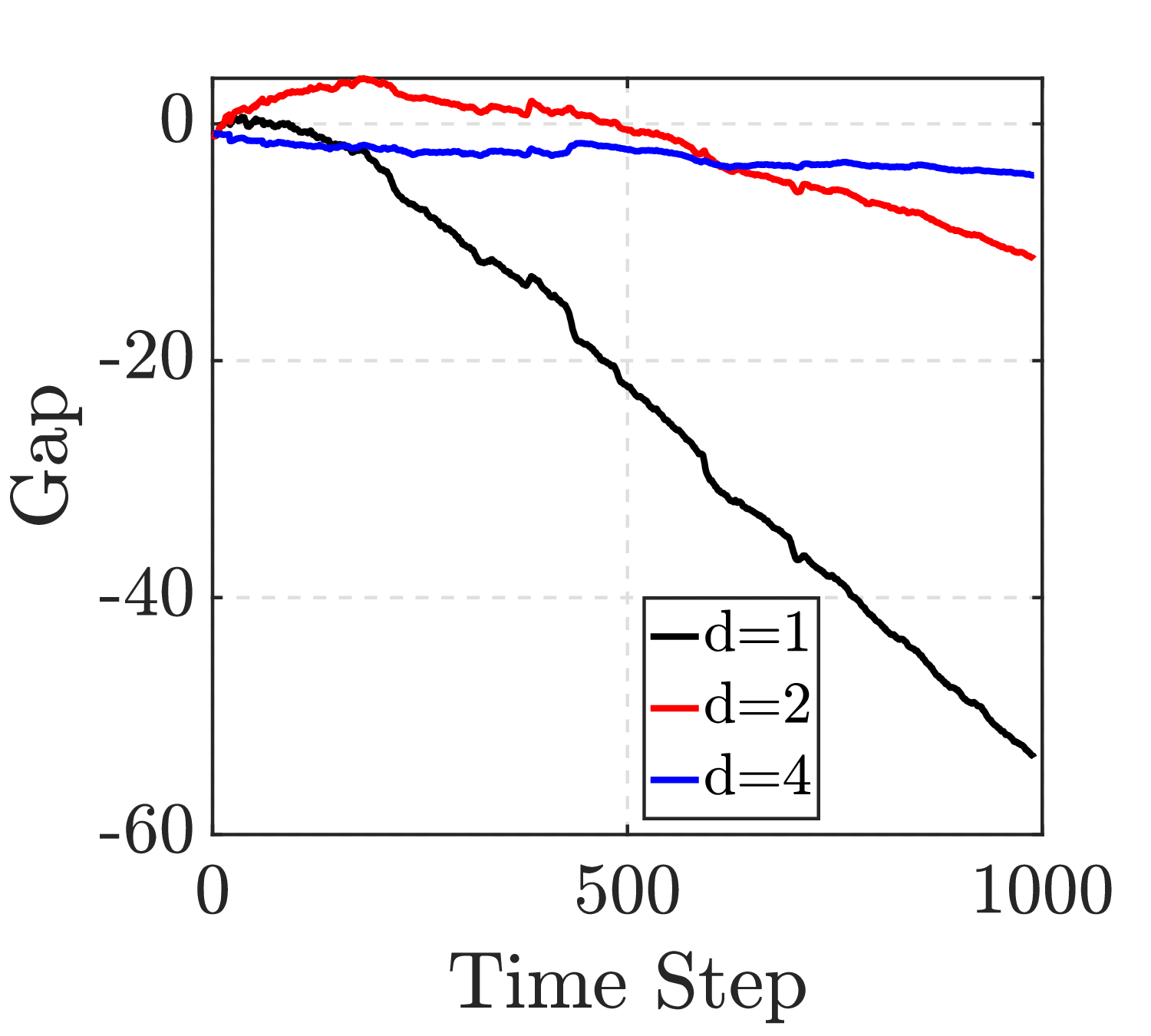}
    \label{figPerformGap}
    } 
    \hspace{1mm}
    \subfigure[Comparison of Different $\beta$.]{\includegraphics[width=0.29\textwidth]{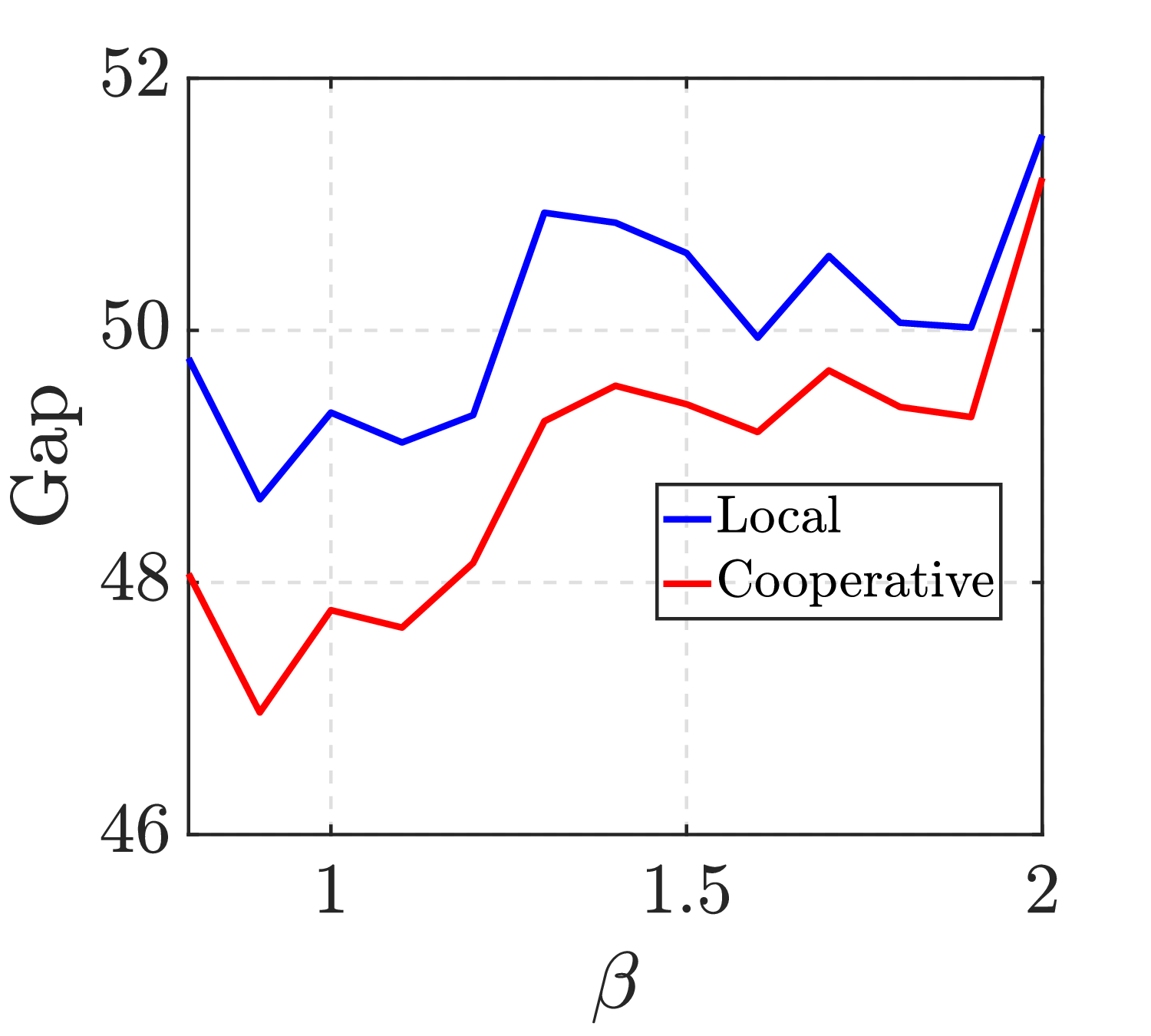}
    \label{figMultiBeta}
    }
    \caption{Numerical experiment with real-life traffic trajectory data.
    }
    
    \label{figRealLife}

    \vspace{-3mm}
\end{figure*}

We then evaluate the {\it performance improvement} using our \texttt{co-Filter} with different time delays $d$. Specifically, we set $\beta=2$ and consider $d=3$ and $d=5$. The results in Figure~\ref{figMultidelay} illustrate the regret between our online predictor that uses delayed external observations and the model-based Kalman predictor that relies solely on local observations. 
In Figure~\ref{figMultidelay}, the orange and green lines represent the performance gap for delay steps of 3 and 5, respectively, while the blue line with $d = \infty$ shows the gap between the online predictor using only local observations and the corresponding local Kalman predictor. This last case aligns with \cite[ Theorem 1]{tsiamis9894660}, which is equivalent to the special case with $d=\infty$ in our analysis. Using delayed external~observations, our \texttt{co-Filter} can outperform the optimal local Kalman predictor, with performance improvement growing linearly with horizon length $N$. As the delay $d$ increases, this benefit diminishes, and the cooperative predictor's performance approaches that of a standard local online predictor.

{\bf Experiment 2: Practical vehicle trajectory prediction.} 
In our second experiment, we use real-world data to evaluate the performance of model-free cooperative prediction. The data is collected from \cite{NGSIMdata}, which contains real-world trajectory data from vehicles with unknown dynamic models. As we cannot directly obtain the sensor data, we use the following method to generate the raw observations, i.e.,
\[
y_k=x_k+v_k,\quad y_k^e=x_k+v_k^e,
\]
where $x_k$ is the trajectory of the vehicle at each time step, and the noise vectors satisfy $v_k\sim \mathcal{N}(0,0.1I), \;v_k^e\sim \mathcal{N}(0,0.1I)$, respectively. The parameter $\beta$ is set to be 1 for the first two experiments, and $T_{\textnormal{init}}$ is set to be 20 for all experiments. 

Our experiment results with traffic data are shown in Figure~\ref{figRealLife}. Specifically, Figure~\ref{figTrajectory} shows the real trajectory and the predicted trajectory with local observations and partially delayed observations. Here, the delay step $d$ is set to be 1. Figure~\ref{figPerformGap} presents the performance improvement of the cooperative method over the local prediction method across varying delay steps. It confirms that cooperative prediction outperforms local prediction, even in the presence of a time delay in real-world data. However, the benefit also diminishes as the delay step increases. In Figure~\ref{figMultiBeta}, we also compare the performance of local prediction and cooperative prediction under different hyperparameter $\beta$ with $d=1$. The cooperative method consistently outperforms the local method across all tested values of $\beta$. Overall, the above experiments verify the effectiveness of utilizing multiple source information for model-free cooperative prediction in real-world traffic scenarios. 

\section{Conclusion}\label{sec: conclusion}

\vspace{-1mm}

In this paper, we addressed the problem of online cooperative prediction for an unknown, non-explosive linear stochastic system. 
To effectively utilize multi-source and partially delayed observations, we have derived an optimal autoregressive model that relates future observations to past delayed ones via conditional distribution theory. We have also developed an online learning algorithm based on this autoregressive model, achieving logarithmic regret with respect to the optimal model-based cooperative predictor. 
Using refined proof techniques, we established a sharper regret bound of $O(\log^3 N)$, which also holds for marginally stable systems. 
Our results highlight the advantage of incorporating multi-source correlated observations into cooperative prediction schemes. Theoretical analysis and numerical experiments demonstrate that with multi-source delayed information, our online cooperative predictor can outperform the optimal model-based predictors that rely solely on local observations.  
Future directions include extending the auto-regressive modeling approach for multi-source data to nonlinear systems, especially for special classes of structured nonlinear models \cite{shang2024willems,shang2026existence}. 

\bibliographystyle{plain}        
\bibliography{autosam}           

\appendix
\noindent 

\section{Proof of Proposition \ref{theorem:optimal-delayed-filter}}\label{Appen: timevarying}
 In this proof, we aim to first provide the expression of the optimal predictor $\hat{y}_{k+1}$ with local information $Y_{0:k}$ and delayed information $Y_{0:k-d}^{\mathrm{e}}$ in the general time-varying form. Then, with the steady-state form of the Kalman filter, we can provide an optimal steady-state predictor with delayed information. From the theory of optimal filtering \cite{kalmanfilter, Andersonoptimal}, the minimal-mean-square-error (MMSE) predictor $\hat{y}_{k+1}$ is equivalent to the expectation of $y_{k+1}$ conditioned on past observations $Y_{0:k}$ and $Y_{0:k-d}^{\mathrm{e}}$, i.e.,
\[
\hat{y}_{k+1}\!=\!\mathbb{E}\left\{y_{k+1}\!\mid\! Y_{0:k},\!Y_{0:k-d}^{\mathrm{e}}\right\}\!=\!C\,\mathbb{E}\!\left\{x_{k+1}\!\mid\! Y_{0:k},Y_{0:k-d}^{\mathrm{e}}\right\},
\]
where the last equality holds for the zero mean of noise $v_{k+1}$. Consider the instrumental variable $\bar{x}_{k+1}\triangleq \mathbb{E}\left\{x_{k+1}\mid Y_{0:k}^{c}\right\} $, i.e., Kalman predictor with non-delayed information $Y_{0:k}^{c}$, then with the Markov property of the linear system, we have
\[
\begin{aligned}
    \hat{x}_{k+1}\triangleq&\mathbb{E}\left\{x_{k+1}\mid Y_{0:k},Y_{0:k-d}^{\mathrm{e}}\right\}
    \\=&\mathbb{E}\left\{x_{k+1}\mid y_{k},\dots,y_{k-d+1},Y_{0:k-d}^{\mathrm{e}}\right\}\\
    =&\mathbb{E}\left\{x_{k+1}\mid y_{k},\dots,y_{k-d+1},\bar{x}_{k-d+1}\right\}.
\end{aligned}
\]
With the conditional distribution theory in \cite{Andersonoptimal}, we have that the optimal predictor $\bar{x}_{k}$ with collected information $Y^c_{0:k-1}$ satisfies the following recursion 
$$
\bar{x}_{k+1}=A\bar{x}_{k}+\bar{L}_{k}\left(y_{k}^c-\bar{C}\bar{x}_{k}\right),$$
where
$\bar{L}_k=A\bar{P}_k\bar{C}^\tr\left(\bar{C}\bar{P}_k\bar{C}^\tr+\bar{R}\right)^{-1}$ and $P_k$ satisfies the following recursion
\begin{equation}\label{eq: RicRecursionApp}
  \bar{P}_k=\Ric\big(A,\bar{C},Q,\bar{R},\bar{P}_{k-1} \big).
\end{equation}
We make the following notation for analysis
\[
\tilde{x}_{k|l}=\mathbb{E}\left\{x_{k+l}\mid y_{k+l-1},\dots,y_{k},\bar{x}_{k}\right\}.
\]
The meaning of $\tilde{x}_{k|l}$ is the conditional mean of $x_{k+l}$ with $\bar{x}_{k}$ and $l$-step local $y_{k},\dots,y_{k+l-1}$observations ahead of time step $k$ and we also have 
$$\hat{x}_{k+1}=\tilde{x}_{k-d+1|d}.$$
For deriving the closed-form expression of $\tilde{x}_{k|l}$, we can directly use the local Kalman filter with initial prediction set as $\tilde{x}_{k|0}=\bar{x}_{k}$. Then update $\tilde{x}_{k|l+1}$ with $\tilde{x}_{k|l}$ and $y_{k+l}$ recursively in the form of local Kalman filter. The specific recursion form can be formulated as
\[
\tilde{x}_{k|1}=(A-L_{k}^{(1)}C)\bar{x}_{k}+L_{k}^{(1)}y_{k},
\]
and
\[
\tilde{x}_{k|l+1}=(A-L_{k}^{(l+1)}C)\tilde{x}_{k|l}+L_{k}^{(l+1)}y_{k+l},
\]
where $L_{k}^{(l)}=AP_{k}^{(l)}C^\tr\left(CP_{k}^{(l)}C^\tr+R\right)^{-1}$ and the $P_{k}^{(l+1)}$ means the prediction error covariance matrix of $\tilde{x}_{k|l}$, i.e., $P_{k}^{(l+1)}=\mathbb{E}\left\{(x_{k+l}-\tilde{x}_{k|l})(x_{k+l}-\tilde{x}_{k|l})^\tr\right\}$. 
It can be computed by the following recursion
\[
P_{k}^{(l+1)}=\Ric(A,C,Q,R,P_k^{(l)}),\;\; P_{k}^{(1)}=\bar{P}_k.
\]
With $\hat{x}_{k+1}=\tilde{x}_{k-d+1|d}$, then we can obtain
\begin{equation}\label{eq: delayed-predictor}
    \hat{x}_{k+1}=  \Phi_{k-d+1,d} \bar{x}_{k-d+1} + \sum_{l=1}^{d}\Phi_{k-d+1,d-l}L_{k-d+1}^{(l)}\;y_{k-d+l},
\end{equation}
where the matrices $\Phi_{.,.}$ are constructed by the following law:
\[
\begin{aligned}
   \Phi_{k-d+1,d-l} &= \prod_{i=l+1}^d (A-L_{k-d+1}^{(i)}C),\quad \Phi_{k-d+1,0}=I. 
\end{aligned}
\]
To sum up, the construction of the optimal predictor $\hat{x}_{k+1}$ with $Y_{0:k}\bigcup Y_{0:k-d}^e$ can be summarized as two steps:
\begin{enumerate}
    \item As we can obtain the global information $Y_{0:k-d}^c$, then for predicting $y_{k+1}$, we can first obtain the centralized optimal predictor for time step $k-d+1$, i.e., $\bar{x}_{k-d+1}$.
    \item From time step  $k-d+1\leq l < k+1$, at each time step we perform the optimal prediction only with $y_{l}$ to obtain $\hat{x}_{k+1}$. 
\end{enumerate}
According to \cite{Andersonoptimal}[Section 4], the recursion $\bar{P}_k$ exponentially converge to the steady-state, i.e., the solution $\bar{P}$ to the following discrete time algebraic Riccati equation 
$$
\bar{P}=A\bar{P}A^\tr+Q-A\bar{P}\bar{C}^\tr\left(\bar{C}\bar{P}\bar{C}^\tr+\bar{R}\right)^{-1}\bar{C}\bar{P}A^\tr.
$$
and global optimal feedback gain $\bar{L}_k$ converges to $\bar{L}=AP\bar{C}^{\tr}\left(\bar{C}\bar{P}\bar{C}^\tr+\bar{R}\right)^{-1}.$
Moreover, we have $ \Phi_{d-l} = \prod_{i=l+1}^d (A-L^{(i)}C)$, with $\Phi_{0}=I$, and $ 
   L^{(l)} \!=\! AP^{(l)}C^\tr\!\left(CP^{(l)}C^\tr\!+\!R\right)^{-1}$ 
with the recursion of $P^{(l)}$ constructed as
\begin{equation}\label{eq: steadyRecursion2}
    P^{(l+1)}=\Ric\big(A,C,Q,R,P^{(l)}\big),\quad P^{(1)}=\bar{P}.
\end{equation}

In the following discussion, to simplify notation, we will focus on the steady-state optimal delayed predictor and how to learn it.

\section{Statistical Property of Linear System and Gram Matrix}
In this section, we introduce some basic properties of the linear stochastic system. Consider the following linear system model
\begin{equation}\label{linearsystem}
\begin{aligned}
x_{k+1}&=Ax_k + \omega_k,\\
y_{k}&=Cx_k + v_{k},\quad k = 0,1,2,\ldots\\
y^{\mathrm{e}}_k&=C^{\mathrm{e}}x_k+v^{\mathrm{e}}_k,
\end{aligned}
\end{equation}
where the meaning of the notations is the same as that presented before. 
For the theoretical analysis of the regret, we present the following basic properties of the linear stochastic system. 

Before we move on, we introduce the following basic Lemma that will be frequently used in the following proofs

\subsection{Basic statistical properties of  the random variables}
We have the following Lemma.
\begin{lemma}\label{lm: HWinequality}
(Hanson-Wright inequality \cite{vershynin2018high}) Let  $X=\left(X_{1}, \ldots, X_{n}\right) \in \mathbb{R}^{n}$  be a random vector with independent, mean-zero, sub-Gaussian coordinates. Let  $\mathcal{A}$  be an  $n \times n$  matrix. Then, for every  $t \geq 0$, we have
    $$
    \begin{aligned}
        \mathbb{P}&\left\{\left|X^{{\tr}} \mathcal{A} X-\mathbb{E} X^{\tr} \mathcal{A} X\right| \geq t\right\} \\
        &\qquad\qquad\leq 2 \exp \left[-c \min \left(\frac{t^{2}}{K^{4}\|\mathcal{A}\|_{F}^{2}}, \frac{t}{K^{2}\|\mathcal{A}\|_2}\right)\right],
    \end{aligned}
    $$
    where $K=\max _{i}\left\|X_{i}\right\|_{\psi_{2}}.$
\end{lemma}
In the above lemma, the norm $\left\|\cdot\right\|_{\psi_2}$ is defined as
\[
\|X\|_{\psi_{2}}=\inf \left\{c>0: \mathrm{E}\left[\exp \left(\frac{X^{2}}{c^{2}}\right)\right] \leq 2\right\}.
\]
In the following proof, we will mainly use the $\psi_2$ norm of a standard Gaussian random variable, which is a constant around $\sqrt{\frac{8}{3}}$.

\begin{lemma}\label{lm: Deviation}
    For any given $\delta\in (0,1)$ and for any Gaussian random vector sequence $X_k$ satisfies $X_k\sim \mathcal{N}(0,I_n)$, we define event $\mathcal{E}_{X}$ as
    \[\mathcal{E}_{X}\triangleq\left\{\left\|X_k\right\|_{2}^2 \leq 2n+3\log \frac{k^2}{\delta}, \;\; \forall k\ge 1\right\},
    \]
    then the event $\mathcal{E}_X$ holds with probability at least $1-\frac{\pi^2 \delta}{6}$.
\end{lemma}
\begin{proof}
     From Lemma 1 in \cite{laurent2000adaptive}, 
     for any Gaussian random vector $X_k\sim \mathcal{N}(0,I_n).$ For each $k\ge 1$, we have
     $$\mathbb{P}\left(\left\|X_k\right\|_2^2\ge n+2\sqrt{n}t+2t^2\right)\leq e^{-t^2}.$$
     Then the following statement also holds
     $$
     \mathbb{P}\left(\left\|X_k\right\|_2^2\ge 2n+3t^2\right)\leq e^{-t^2}.
     $$
     For each time step $k$, take $t=\sqrt{\log\frac{k^2}{\delta}}$, then we have$$\mathbb{P}\left(\left\|X_k\right\|_2^2\leq 2n+3\log\frac{k^2}{\delta}\right)\ge 1- \frac{\delta}{k^2}.
     $$

     Take a union bound over all $k$, we have
     $$
     \mathbb{P}\Big(\!\left\|X_k\right\|_2^2\leq 2n+3\log\frac{k^2}{\delta}, \forall k\in\mathbb{N}\!\Big)\!\ge 1- \sum_{k=1}^{\infty}\frac{\delta}{k^2}\!=\!1-\frac{\pi^2\delta}{6},
     $$
     where the inequality utilizes the facts $\mathbb{P}(\mathcal{E}_1 \cap \mathcal{E}_2)=\mathbb{P}(\mathcal{E}_1)+\mathbb{P}(\mathcal{E}_2)-\mathbb{P}(\mathcal{E}_1 \cup \mathcal{E}_2)$ and $\mathbb{P}(\mathcal{E}_1 \cup \mathcal{E}_2)\leq 1$.
     \end{proof}
The above lemma illustrates that for a zero-mean Gaussian random variable, we can find an increasing upper bound with the same order as $O(\log k)$ that uniformly holds for the deviation between the random variable and its corresponding variance, with high probability $1-\delta$. This lemma will be repeatedly used in the analysis of the regret bound.

\subsection{Upper bound of $\bar{Z}_{k,p+d}\bar{Z}_{k,p+d}^{\tr}$ with high probability}\label{upperbound}
The following lemma provides an upper bound for $\bar{Z}_{k,p+d}\bar{Z}_{k,p+d}^{\tr},\forall k\ge T_{\text{init}}$.
\vspace{6pt}

\begin{lemma}\label{lm: upperbound}
    For a given probability $\delta$, if $p$ satisfies $p\leq \beta \log k$ for each  $k\ge T_{\textnormal{init}}$, we have 
    $$\mathbb{P}\left\{\bar{Z}_{k}\bar{Z}_{k}^{\tr}\!\leq\! \operatorname{poly}\!\left(d,\beta,\log\frac{1}{\delta}\!\right)k^{2\kappa+1} I,\forall k\ge T_{\textnormal{init}}\right\}\ge 1-\delta,$$ where $\operatorname{poly}(\cdot)$ means polynomial operator of the elements.
\end{lemma}
\begin{proof}
   For the expectation of Gram matrix $Z_{k,p+d}Z_{k,p+d}^{\tr}$, denote $\Gamma_{k,p+d}^Z=\mathbb{E}\left\{Z_{k,p+d}Z_{k,p+d}^{\tr}\right\}$, then we have
$$
\begin{aligned}
    \left\|\Gamma_{k,p+d}^Z\right\|_2^2\leq&\mathbb{E}\left\{Z_{k,p+d}^{\tr}Z_{k,p+d}\right\}\\
    =&\text{trace}\Big(\sum_{i=k-d}^{k-1}\!\mathbb{E}\left\{y_{i}y_{i}^{\tr}\right\}+\!\!\sum_{i=k-d-p}^{k-d-1}\!\mathbb{E}\left\{y_{i}^{c}y_{i}^{(c)\tr}\right\}\Big)\\
    \leq&d\text{trace}(R)+p\text{trace}(\bar{R})
    \\&+n\left\|\bar{C}\right\|_2^2\left\|Q\right\|_2\sum_{i=k-p-d}^{k-1}\sum_{l=0}^{i-1}\left\|A^l\right\|_2^2.
\end{aligned}
$$
Note that exists $M_1$, such that $\left\|A^l\right\|_2\leq M_1l^{\kappa-1}\rho(A)^{l}$, where $\kappa$ is the order of the largest Jordan block of  matrix $A$. Then we have 
$$
\mathbb{E}\left\{Z_{k,p+d}^TZ_{k,p+d}\right\}\leq M(p+d)k^{2\kappa-1},\quad \forall k\in\mathbb{N},
$$
where $M$ is a constant only related to system parameters.
With Lemma \ref{lm: Deviation}, we have $Z_{k}Z_{k}^{{\tr}}\leq M(p+d)\left(2m_{d,p}+3\log\frac{k^2}{\delta}\right)k^{2\kappa-1} I$ holds uniformly for all $k$ with probability at least $1-\frac{\pi^2\delta}{6}$.
Together with the selection rule of $p\leq \beta \log k$, we can obtain 
$$
\begin{aligned}
    \bar{Z}_{k,p+d}\bar{Z}_{k,p+d}^{\tr}=&\sum_{l=p+d}^{k}Z_{l,p+d}Z_{l,p+d}^{\tr}\\
    \leq &\operatorname{poly}\left(d,\beta,\log\frac{1}{\delta}\right)k^{2\kappa
+1}I
\end{aligned}
$$
uniformly holds for all $k$ with probability $1-\frac{\pi^2 \delta}{6}$. We now complete the proof. 
\end{proof}

\subsection{Almost Sure Persistent excitation for $\bar{Z}_{k,p+d}\bar{Z}_{k,p+d}^{\tr}$ with high probability}
In this subsection, we will provide a condition for the persistent excitation of $\bar{Z}_{k,p+d}\bar{Z}_{k,p+d}^{\tr}$, i.e., find a $k_0$ such that $\bar{Z}_{k,p+d}\bar{Z}_{k,p+d}^{\tr}\ge c_0kI,\;\forall k\ge k_0$ with high probability. The property of persistent excitation is of vital importance to guarantee the performance of the online prediction Algorithm~\ref{algPrediction}.

\vspace{6pt}

\begin{lemma}\label{lm: propersistent}
For a given system model \eqref{eq: LinearSystem} and \eqref{eq:external-source}, then for any fixed failure probability $\delta$, horizon parameter $\beta$ and delay step $d$, there exists a number $k_0=\operatorname{poly}\left(d,\beta,\log\frac{1}{\delta}\right)$, such that for any $k>k_0$ and $p\leq\beta \operatorname{log} k$, there is $\bar{Z}_{k,p+d}\bar{Z}_{k,p+d}^{\tr}\ge \frac{\sigma_{\bar{R}}}{4}kI$ uniformly for all $k$ with probability $1-\delta$, where $\sigma_{\bar{R}}$ is the smallest eigenvalue of matrix $\bar{R}$.
\end{lemma}
\begin{proof}
    Before we move on, to simplify the notation, we denote $m_{p,d}=\bar{m}p+md$ as the dimension of the sample $Z_{k,p+d}$ for given $p$ and $d$, which will also be used in the proof of Theorem~\ref{thm: regret}. We first consider the structure of $\bar{Z}_{k,p+d}$, with the expression of $y_{k}$, we can divide the collected matrix $\bar{Z}_{k,p+d}$ in to the following form
    \[
    \bar{Z}_{k, p+d}=\bar{Y}_{k, p+d}+\tilde{V}_{k, p+d},
    \]
where
$
\bar{Y}_{k, p+d}\triangleq \left[\begin{smallmatrix}
        Cx_{p+d-1} & \cdots& Cx_{k-1} \\
\vdots & & \vdots\\
Cx_{p} & & Cx_{k-d} \\
\bar{C}x_{p-1} & \cdots  & \bar{C}x_{k-d-1} \\
\vdots & & \vdots\\
\bar{C}x_{0} &\cdots &\bar{C}x_{k-p-d} \\
    \end{smallmatrix}\right]$ denotes the system state part, and
    $
    \tilde{V}_{k, p+d}\triangleq\left[ \begin{smallmatrix}
        v_{p+d-1} & \cdots& v_{k-1} \\
\vdots & & \vdots\\
v_{p} & & v_{k-d} \\
v_{p-1}^{c} & \cdots  & v_{k-d-1}^{c} \\
\vdots & & \vdots\\
v_{0}^{c} &\cdots &v_{k-p-d}^{c} \\
    \end{smallmatrix}\right]
$ denotes the noise part.
Then we have
\[
\begin{aligned}
    \bar{Z}_{k, p+d}\bar{Z}_{k, p+d}^\tr=\underbrace{\bar{Y}_{k, p+d}\bar{Y}_{k, p+d}^\tr+\tilde{V}_{k, p+d}\tilde{V}_{k, p+d}^\tr}_{\text{quadratic term}}\qquad\qquad\\+\underbrace{\bar{Y}_{k, p+d}\tilde{V}_{k, p+d}^\tr+\tilde{V}_{k, p+d}\bar{Y}_{k, p+d}^\tr}_{\text{cross term}}.
\end{aligned}
\]
Motivated by \cite{tsiamis9894660}, we will show that for sufficiently large $k\ge\poly(d,\beta,\log\left(\frac{1}{\delta}\right))$, the cross term $\mathcal{C}_{k,p+d}\triangleq\bar{Y}_{k, p+d}\tilde{V}_{k, p+d}^\tr\!+\!\tilde{V}_{k, p+d}\bar{Y}_{k, p+d}^\tr$ will be dominated by the quadratic term uniformly with high probability $1\!-\!\delta$, i.e.,
\[
\left|u^\tr\mathcal{C}_{k,p+d}u\right|\leq\frac{1}{2}u^\tr\left(\bar{Y}_{k, p+d}\bar{Y}_{k, p+d}^\tr+\tilde{V}_{k, p+d}\tilde{V}_{k, p+d}^\tr\right)u
\]
holds for all vector $u^\tr u=1$ with proper dimension.
Together with that, the quadratic term satisfies
\[
\bar{Y}_{k, p+d}\bar{Y}_{k, p+d}^\tr+\tilde{V}_{k, p+d}\tilde{V}_{k, p+d}^\tr\ge\tilde{V}_{k, p+d}\tilde{V}_{k, p+d}^\tr\ge \frac{\sigma_{\bar{R}}}{2}kI.
\]
with high probability $1-\delta$.
We will have
\[
\bar{Z}_{k, p+d}\bar{Z}_{k, p+d}^\tr\ge \frac{\sigma_{\bar{R}}}{4}kI 
\]
with high probability $1-2\delta$.

For the first step, we provide the condition for guaranteeing 
$
\tilde{V}_{k, p+d}\tilde{V}_{k, p+d}^\tr\ge \frac{\sigma_{\bar{R}}}{2}kI
$. We first consider the matrix 
\[\tilde{R}_{p+d}\triangleq \operatorname{diag}\Big(\underbrace{R,\dots,R}_d,\underbrace{\bar{R},\dots,\bar{R}}_p\Big).
\]
Then we have $\tilde{R}_{p+d}\ge \sigma_{\bar{R}} I$ and the matrix $\tilde{R}_{p+d}^{-\frac{1}{2}}\tilde{V}_{k,p+d}$ is composed with i.i.d Gaussian random variable.
With Lemma I.5 in \cite{tsiamis9894660}, for a fixed probability $\delta$, there exists a universal $c_1$ such that if
$
k \geq c_1 m_{p,d} \log (m_{p,d} / \delta),
$
there is
\begin{equation}\label{eq: noisePE}
    \frac{1}{2}kI\leq \tilde{R}_{p+d}^{-\frac{1}{2}}\tilde{V}_{k,p+d}\tilde{V}_{k,p+d}^\tr\tilde{R}_{p+d}^{-\frac{1}{2}}\leq \frac{3}{2}kI,
\end{equation}
Note that the above condition holds for one specific $k$; to obtain a uniform condition for all $k$, we can choose the probability as $\frac{\delta}{k^2}$ for each time step $k$ to obtain a uniform bound. Then for a fixed $\delta$, if
\begin{equation}\label{eq: PEcondition1}
    k \geq c_1 m_{p,d} \log (m_{p,d}k^2 / \delta)
\end{equation}
for each $k$, then the result \eqref{eq: noisePE} holds for all $k$
with probability at least $1-\frac{\pi^2}{6}\delta$. For condition \eqref{eq: PEcondition1}, note that $p\leq \beta\log k$, then the above condition can be relaxed to
$k \geq c_1 (\beta \bar{m}\log k+dm) \log ((\beta \bar{m}\log k+dm)k^2 / \delta)$.
We can further simplify the above condition as \[
k \geq \poly(d,\beta) \log k \left(\log k + \log\frac{1}{\delta}\right).
\]
With the property of $\log$ function, we have that for given $\beta, d,\delta$, there exists $k_1 = \poly\left(\beta, d,\log\left(\frac{1}{\delta}\right)\right)$, for any $k\ge k_1$, the condition \eqref{eq: noisePE} holds with probability $1-\frac{\pi^2}{6}\delta$. In the following analysis, we directly consider the case with $k\ge k_1$.

For the cross term $\mathcal{C}_{k,p+d}$, for any vector $u$ with $u^
\tr u=1$, we have
\[
\begin{aligned}
    &\left|\!u^\tr\! \mathcal{C}_{k,p+d} u\!\right|
    \!\!\leq \!2\!\left\|\!\left(\!\lambda I\!+\!\bar{Y}_{k,p+d} \bar{Y}_{k,p+d}^{\tr}\!\right)^{\!-\!\frac{1}{2}}\! \!\bar{Y}_{k,p+d} \!\tilde{V}_{k,p+d}^{\tr} \!\tilde{R}_{p+d}^{-\frac{1}{2}}\right\|_{2}\\&
    \times \sqrt{u^{\tr}\left(\lambda I+\bar{Y}_{k,p+d} \bar{Y}_{k,p+d}^{\tr}\right) u}\;\; \sqrt{u^{\tr} \tilde{R}_{p+d} u}\\
    &\leq \sqrt{\frac{2}{k}}\left\|\left(\lambda I+\bar{Y}_{k,p+d} \bar{Y}_{k,p+d}^{\tr}\right)^{-\frac{1}{2}} \bar{Y}_{k,p+d} \tilde{V}_{k,p+d}^{\tr} \tilde{R}_{p+d}^{-\frac{1}{2}}\right\|_{2}\\&
    \times \left(u^{\tr}\left(\lambda I+\bar{Y}_{k,p+d} \bar{Y}_{k,p+d}^{\tr}\right) u + \frac{k}{2}u^{\tr} \tilde{R}_{p+d} u\right)\\
    &\leq \sqrt{\frac{2}{k}}\left\|\left(\lambda I+\bar{Y}_{k,p+d} \bar{Y}_{k,p+d}^{\tr}\right)^{-\frac{1}{2}} \bar{Y}_{k,p+d} \tilde{V}_{k,p+d}^{\tr} \tilde{R}_{p+d}^{-\frac{1}{2}}\right\|_{2}\\&
    \times \left(u^{\tr}\left(\lambda I+\bar{Y}_{k,p+d} \bar{Y}_{k,p+d}^{\tr}\right) u + u^{\tr} \tilde{V}_{k,p+d}\tilde{V}_{k,p+d}^\tr u\right),
\end{aligned}
\]
where the last inequality is from condition \eqref{eq: noisePE}. 
Then we only need to show that there exists $k_0=\poly\left(d,\beta,\log\left(\frac{1}{\delta}\right)\right)$, such that for any $k\ge k_0$, there is \[\sqrt{\frac{2}{ k}}\left\|\left(\lambda I+\bar{Y}_{k,p+d} \bar{Y}_{k,p+d}^{\tr}\right)^{-\frac{1}{2}} \bar{Y}_{k,p+d} \tilde{V}_{k,p+d}^{\tr} \tilde{R}_{p+d}^{-\frac{1}{2}}\right\|_{2}\leq \frac{1}{2}\]
with high probability. Denote
\[
M_{k,p+d}\triangleq \bar{Y}_{k,p+d}^{\tr}\left(\lambda I+\bar{Y}_{k,p+d} \bar{Y}_{k,p+d}^{\tr}\right)^{-1} \bar{Y}_{k,p+d},
\] then the above inequality condition is also equivalent to
\begin{equation}\label{eq: MiddleCondition1}
    \tilde{R}_{p+d}^{-\frac{1}{2}}\tilde{V}_{k,p+d}M_{k,p+d}\tilde{V}_{k,p+d}^{\tr} \tilde{R}_{p+d}^{-\frac{1}{2}}\leq \frac{k}{8}I.
\end{equation}
Note that each row of matrix $\tilde{R}_{p+d}^{-\frac{1}{2}}\tilde{V}_{k,p+d}$ is i.i.d Gaussian, it is easy to verify the following facts: 
\[
\left\|M_{k,p+d}\right\|_2\leq 1,\quad \operatorname{trace}\left(M_{k,p+d}\right) \leq m_{p,d}.
\]
Moreover, we also have
\[
\left\|M_{k,p+d}\right\|_F^2\!\leq\! \operatorname{trace}\! \left(\!\left(\lambda I+\bar{Y}_{k,p+d} \bar{Y}_{k,p+d}^{\tr}\right)^{-1} \bar{Y}_{k,p+d}\bar{Y}_{k,p+d}^{\tr}\right),\]
then there is $\left\|M_{k,p+d}\right\|_F^2\leq m_{p,d}$. For each $k$, we choose 
\[t=K^2\left(\frac{1}{2}\left\|M_{k,p+d}\right\|_F^2+\frac{1}{c}\log\left(\frac{m_{p,d}k^2}{\delta}\right)\right),
\] 
where $K$ is the $\psi_2$ norm of standard Gaussian distribution and $c$ is the constant in Lemma \ref{lm: HWinequality} (Hanson-Wright inequality). Then we can verify that 
$
\frac{ct^2}{K^4\left\|M_{k,p+d}\right\|_F^2}\ge\log\left(\frac{m_{p,d}k^2}{\delta}\right)
$
and
$
\frac{ct}{K^2\left\|M_{k,p+d}\right\|_2}\ge\log\left(\frac{m_{p,d}k^2}{\delta}\right).
$
We denote the $i$-th row of matrix $\tilde{R}_{p+d}^{-\frac{1}{2}}\tilde{V}_{k,p+d}$ as $g_i^\tr$, i.e., $\tilde{R}_{p+d}^{-\frac{1}{2}}\tilde{V}_{k,p+d}=\begin{bmatrix}
    g_1,g_2,\dots,g_{\bar{m}p+md}
\end{bmatrix}^\tr$, then for each $i$, $g_i$ is a vector with i.i.d standard Gaussian random variable, therefore we have
\[
\mathbb{E}\left\{g_i^\tr M_{k,p+d}g_i\right\}=\operatorname{trace}(M_{k,p+d})\leq m_{p,d}.
\]
With Lemma \ref{lm: HWinequality}, we have
\[
\begin{aligned}
    \mathbb{P}\left(g_i^\tr M_{k,p+d}\;g_i\ge \mathbb{E}\left\{g_i^\tr M_{k,p+d}\;g_i\right\}+t \right)\qquad\qquad\qquad\\
    \leq 2e^{-\log \frac{m_{p,d}k^2}{\delta}}=\frac{2\delta}{m_{p,d}k^2}.
\end{aligned}
\]
Take a union over all $i=1,\dots, m_{p,d}$, we have 
$\operatorname{trace}\big(\tilde{R}_{p+d}^{-\frac{1}{2}}\tilde{V}_{k,p+d}M_{k,p+d}\tilde{V}_{k,p+d}^\tr\tilde{R}_{p+d}^{-\frac{1}{2}}\big)\leq m_{p,d}(m_{p,d}+t)$ holds with probability at least $1- \frac{2\delta}{k^2}$.
Together with the fact that 
$
\lambda_{\max}\left(P\right)\leq \operatorname{trace}\left(P\right)
$, where $P$ is positive semidefinite. Then we have the fact
\[
\mathbb{P}\!\Big(\!\tilde{R}_{p+d}^{-\frac{1}{2}}\!\tilde{V}_{k,p+d}\!M_{k,p+d}\!\tilde{V}_{k,p+d}^\tr\tilde{R}_{p+d}^{-\frac{1}{2}}\!\!\leq\!\! m_{p,d}(\!m_{p,d}\!+t\!)\!I\! \Big)\!\!\ge \!\!1\!-\frac{\pi^2\delta}{3}
\]
uniformly for all $k$.
Consider the term $B(p,t)\triangleq m_{p,d}(m_{p,d}+t)$, note that $p\leq \beta\log k$, together with the expression of $t$, then we can also relax the term $B(p,t)$ with
\[
B(p,t)\leq \poly\left(d,\beta\right)\log k\left(\log k+\log\left(\frac{1}{\delta}\right)\right).
\]
For the condition \eqref{eq: MiddleCondition1}, we can relax it by letting 
\[
\poly\left(d,\beta\right)\log k\left(\log k+\log\left(\frac{1}{\delta}\right)\right)\leq \frac{k}{8}.
\]
Therefore we can obtain that there exists $k_0=\poly\left(d, \beta,\log\left(\frac{1}{\delta}\right)\right)$, such that for any $k\ge k_0$, there is 
\[
\tilde{R}_{p+d}^{-\frac{1}{2}}\tilde{V}_{k,p+d}M_{k,p+d}\tilde{V}_{k,p+d}^\tr\tilde{R}_{p+d}^{-\frac{1}{2}}\leq \frac{k}{8} I
\]
with probability at least $1- \frac{2\pi^2\delta}{6}$. Then this lemma is proved.
\end{proof}

\section{Proof of Theorem \ref{thm: regret}}\label{Appen: RegretProof}

We provide all the proofs for Theorem~\ref{thm: regret} here.  
We first present the regression model here for reference
\[
y_{k+1}=G_{p+d}Z_{k+1,p+d}+C \Phi_d(A-\bar{L} \bar{C})^{p} \bar{x}_{k-d-p}+r_{k+1}.
\]
To simplify the proof, we denote $b_{k,p+d}=\Phi_d(A-\bar{L} \bar{C})^{p} \bar{x}_{k-d-p}$ as the bias term.  

\subsection{General Analysis}
For the regret $\mathcal{R}_N$, we first have 
\begin{align} \label{eq:decomposition-regret}
    \mathcal{R}_{N} \triangleq& \sum_{k=T_{\text {init }}}^{N}\left\|y_{k}-\tilde{y}_{k}\right\|_2^{2}\;-\sum_{k=T_{\text {init }}}^{N}\left\|y_{k}-\hat{y}_{k}\right\|_2^{2} \nonumber \\
    =&\underbrace{\sum_{k=T_{\text {init }}}^{N}\left\|\hat{y}_{k}-\tilde{y}_{k}\right\|_{2}^{2}}_{\mathcal{L}_{N}}\;+\;2 \underbrace{\sum_{k=T_{\text {init }}}^{N} r_{k}^{\tr}\left(\hat{y}_{k}-\tilde{y}_{k}\right)}_{\text {martingale term }} .
\end{align}
It is shown in \cite[Theorem 1]{tsiamis9894660} that the martingale term is with the same order as $\tilde{\mathcal{O}}(\sqrt{\mathcal{L}_N})$, hence we mainly need to analyze the term $\mathcal{L}_N$ to provide a bound for $\mathcal{R}_N$. 
With the regression model, we can rewrite the $\tilde{y}_{k+1}$ as
$$\begin{aligned}
\tilde{y}_{k+1}= & \left(\sum_{l=p+d}^{k}\left(G_{p+d} Z_{l, p+d}+b_{l, p+d}+r_{l}\right) Z_{l, p+d}^{{\tr}}\right) \\
&\times V_{k, p+d}^{-1} \!Z_{k+\!1, p\!+\!d} \\
=\!G_{p+d}&\!\left(\!I\!-\!\!\lambda V_{k, p+d}^{-1}\!\right)\!\! Z_{k+1, p+d}\!+\!\!\!\!\!\sum_{l=p\!+\!d}^{k} \!\!\!r_{l} Z_{l, p+d}^{{\tr}}\!V_{k, p+d}^{-1}\! Z_{k+1, p+d}\\&+\sum_{l=p+d}^{k} b_{l, p+d} Z_{l, p+d}^{{\tr}} V_{k\!, p+d}^{-1}\!Z_{k+1,\!p+d}.
\end{aligned}$$
Hence we have
$$
\begin{aligned}
\tilde{y}_{k+1}-\hat{y}_{k+1}= & -\lambda G_{p+d}  V_{k, p+d}^{-1} Z_{k+1, p+d}\\+\sum_{l=p+d}^{k}& b_{l, p+d} Z_{l, p+d}^{{\tr}}V_{k, p+d}^{-1} Z_{k+1, p+d}-b_{k+1, p+d}\\
+\sum_{l=p+d}^{k}&r_{l} Z_{l, p+d}^{{\tr}} V_{k, p+d}^{-1} Z_{k+1, p+d}.
\end{aligned}
$$
In the following analysis, inspired by the proof of \cite[Theorem 1]{tsiamis9894660}, we assume that the parameter $\beta$ has already been chosen to satisfy the requirement provided in Theorem~\ref{thm: regret}. Then, for each time step $k$, the parameter $p$ is fixed.  {\bf In order to simplify the notations, we will eliminate the subscript $p+d$ if unnecessary, for example, $Z_{k,p+d}$ will be replaced by $Z_k$.}
The difference at time step $k+1$ can be formulated as
$$
\begin{aligned}
\left\|\tilde{y}_{k+1}-\hat{y}_{k+1}\right\|_{2}^{2} \leq  6\left(\left\| \lambda G_{p+d} V_{k}^{-\frac{1}{2}}\right\|_{2}^{2}+\left\| B_{k} \bar{Z}_{k}^{{\tr}} V_{k}^{-\frac{1}{2}} \right\|_{2}^{2} \right.\\
\left.+\left\|\mathcal{R}_{k} \bar{Z}_{k}^{{\tr}} V_{k}^{-\frac{1}{2}}\right\|_{2}^{2}\right)\times\left\|V_{k}^{-\frac{1}{2}} Z_{k+1}\right\|_{2}^{2}+2\left\|b_{k+1}\right\|_{2}^{2}.
\end{aligned}$$ 
On considering that the parameter $p$ varies with the time step $k$, we also need to divide the gap term $\mathcal{L}_N$ into different epochs. Denote the number of all the epoch is $N_E$, i.e., $N=2^{N_E}T_{\text{init}}$, and
let $T_l = 2^{l-1}T_{\text{init}}+1$ be the beginning time step of the $l$-th epoch, then we have
$$
\begin{aligned}
\mathcal{L}_{N}=&\sum_{l=1}^{N_E}\sum_{k=T_l}^{2T_l-2}\left\|\tilde{y}_{k+1}-\hat{y}_{k+1}\right\|_{2}^{2}\\
\leq&\sum_{l=1}^{N_E}\sup_{T_l\leq k\leq 2T_l-2
}\!\!6\left(\left\| \lambda G_{p+d} V_{k}^{-\frac{1}{2}}\right\|_{2}^{2}\!\!+\!\left\| B_{k} \bar{Z}_{k}^{{\tr}} V_{k}^{-\frac{1}{2}} \right\|_{2}^{2} \right. \\
&\left. +\left\|\mathcal{R}_{k,p+d} \bar{Z}_{k}^{{\tr}} V_{k}^{-\frac{1}{2}}\right\|_{2}^{2}\right)\times\sum_{k=T_l}^{2T_l-1}\left\|V_{k}^{-\frac{1}{2}} Z_{k+1}\right\|_{2}^{2}\\&+2 \sum_{l=1}^{N_E}\sum_{k=T_l}^{2T_l-2}\left\|b_{k+1}\right\|_{2}^{2}.\\
\end{aligned}
$$
In the following analysis, we will gradually analyze the above terms from a new perspective and then prove the logarithm bound provided in Theorem~\ref{thm: regret} with high probability $1-\delta$. 

\subsection{Proof of Lemma \ref{lemma:bias-factor}: Uniform Boundedness of Bias error and the selection of parameter $\beta$}\label{sectionBeta}

Note that the bias error are in two parts, i.e., the bias regression error $\left\| B_{k} \bar{Z}_{k}^{{\tr}} V_{k}^{-\frac{1}{2}} \right\|_{2}^{2}$ and bias accumulation error $\sum_{k=T_l}^{2T_l-1}\left\|b_{k+1}\right\|_{2}^{2}$ for each epoch. We first consider the first term, note that at each epoch $l$, for any $T_l\leq k\leq 2T_l-2$, there is
$$
\begin{aligned}
    \left\| B_{k} \bar{Z}_{k}^{{\tr}} V_{k}^{-\frac{1}{2}} \right\|_{2}^{2}=& \left\|B_{k} \bar{Z}_{k}^{{\tr}} V_{k}^{-1}\bar{Z}_{k} B_{k}^{\tr}\right\|_2^2
    \leq \sum_{l=p+d}^{k}\left\|b_{l}\right\|_{2}^{2}\\
    \leq&\sum_{l=p+d}^{k}\left\|C\Phi_d(A-\bar{L}\bar{C})^p\right\|_2^2\left\|\bar{x}_{l-p-d}\right\|_2^2.
\end{aligned}
$$
where $p=\beta \log(T_l-1)$.
Denote $\Gamma_k=\mathbb{E}\left\{\bar{x}_k\bar{x}_k^\tr\right\}$, with Lemma \ref{lm: Deviation}, we have for a fixed probability $\delta$, there is
$$
\mathbb{P}\!\left(\!\left\|\bar{x}_{k}\right\|_2^2\!\leq\!\Big(2n\!+\!3\log \frac{k^2}{\delta}\!\Big)\left\|\Gamma_k\right\|_2^2,\forall k>T_{\textnormal{init}}\right)\!\ge\! 1-\frac{\pi^2\delta}{6}.
$$
Note that
\[
\left\|\Gamma_k\right\|_2^2\leq \left\|Q\right\|_2\sum_{i=0}^{k-1}\left\|A^i\right\|_2^2.
\]
Similar to Lemma \ref{lm: upperbound}, we have that
$$
\mathbb{P}\left(\left\|\bar{x}_{k}\right\|_2^2 \leq \poly\left(\log\frac{1}{\delta}\right) k^{2\kappa},\;\; \forall k\ge T_{\text{init}}\right)>1-\frac{\pi^2\delta}{6}.
$$
Note that $\left\|C\Phi_d(A-\bar{L}\bar{C})^p\right\|_2^2\leq M\rho(A-\bar{L}\bar{C})^p\rho_0^d$. 
We choose $\beta=\frac{M_3}{\log(1/\rho(A-\bar{L}\bar{C}))}$, then we have
$$
\begin{aligned}
    \left\| B_{k} \bar{Z}_{k}^{{\tr}} V_{k}^{-\frac{1}{2}} \right\|_{2}^{2}\leq& \poly\left(\log\frac{1}{\delta}\right)(2T_l)^{2\kappa+1}\rho_0^d\\&\times\rho(A-\bar{L}\bar{C})^{\frac{M3}{\log(1/\rho(A-\bar{L}\bar{C}))}\log(T_l-1)}\\
    \leq&\poly\left(\log\frac{1}{\delta}\right)(2T_l)^{2\kappa+1}\frac{\rho_0^d}{(T_l-1)^{M3}}.
\end{aligned}
$$
Then we only need to choose $M_3=2\kappa+1$, then there is
$$
\left\| B_{k} \bar{Z}_{k}^{{\tr}} V_{k}^{-\frac{1}{2}} \right\|_{2}^{2}\leq\sum_{l=p+d}^{2T_l-2}\left\|b_{l}\right\|_{2}^{2}\leq\operatorname{poly}\left(\log\frac{1}{\delta}\right).
$$
Note that the above bound holds for all $T_l\leq k\leq 2T_l-2$, together with the arbitrariness of $l$, we have 
$$
\sup_{1\leq l\leq N_E}\sup_{T_l\leq k\leq 2T_l-2}\left\| B_{k} \bar{Z}_{k}^{{\tr}} V_{k}^{-\frac{1}{2}} \right\|_{2}^{2}\leq\operatorname{poly}\left(\log\frac{1}{\delta}\right)
$$
with high probability $1-\frac{\pi^2\delta}{6}$
Moreover, for the term  $\sum_{l=1}^{N_E}\sum_{k=T_l}^{2T_l-2}\left\|b_{k+1}\right\|_{2}^{2}$, we have
$$
\sum_{l=1}^{N_E}\sum_{k=T_l}^{2T_l-2}\!\!\left\|b_{k+1}\right\|_{2}^{2}\!\leq\!\! \frac{\log(N/T_{\text{init}})}{\log2} \operatorname{poly}\!\Big(\!\log\frac{1}{\delta}\Big)\!=\!\mathcal{O}(\log N).
$$
Therefore it is sufficient to guarantee the uniform boundedness of bias error only with the parameter $\beta$ chosen to be proportional to $1/\log\rho(A-\bar{L}\bar{C})$ and the order $\kappa$ of the marginal stable Jordan block of $A$, which is independent to the time step $k$ or the initial time $T_{\text{init}}$.

\subsection{Persistent Excitation of the Gram matrix $\bar{Z}_{k }\bar{Z}_{k }^{\tr}$ for all $k\ge T_{\text{init}}$}
Consider the event
$$
\mathcal{E}_{PE} \triangleq\left\{\bar{Z}_{k}\bar{Z}_{k}^{\tr}\ge \frac{\sigma_{\bar{R}}}{4}kI,\;\; \forall k\ge T_{\text{init}}\right\}.$$
From Lemma \ref{lm: propersistent}, we can obtain that for the above fixed selecting policy $\beta$, time delay step $d$ and probability $\delta$, there exists a number $k_0=\allowbreak\operatorname{poly}\Big(d,\beta,\allowbreak \log\left(\frac{1}{\delta}\right)\Big),$ such that the event $\mathcal{E}_{PE}$ holds for all $k>k_0$ uniformly with probability at least $1-\delta$. Without loss of generality, we can assume that the length of the warm-up horizon $T_{\text{init}}$ satisfies $T_{\text{init}}>k_0$. Then the persistent excitation condition holds for all $k\ge T_{\text{init}}$ with probability $1-\delta$. This condition will be repeatedly used in the following analysis.

\subsection{Uniform boundedness of $\left\| \lambda G_{p+d}V_{k}^{-\frac{1}{2}}\right\|_{2}^{2}$ with high probability}
For the regularization error term $\left\| \lambda G_{p+d} V_{k}^{-\frac{1}{2}}\right\|_{2}^{2}$, we have the following inequality
$$
\begin{aligned}
    \lambda G_{p+d}\!V_{k}^{-\frac{1}{2}}\!\!\left(\!\lambda G_{p+d}V_{k}^{-\frac{1}{2}}\!\right)^{\tr}\!\!\!\!\!=&\lambda^2\!G_{p+d} \!V_{k}^{-1}\!G_{p+d}^{\tr}\!\!\leq\!\lambda G_{p+d}\!G_{p+d}^{\tr},
\end{aligned}
$$
with the expression of $G_{p+d}$ as $G_{p+d}=\begin{bmatrix}
    G_{p+d}^{(1)}&G_{p+d}^{(2)}
\end{bmatrix}$ and the two parts take the expression as
\[
\begin{aligned}
G_{p+d}^{(1)}\triangleq&\begin{bmatrix}
    C L^{(d)}& \cdots& C \Phi_{d-1}L^{(1)}
\end{bmatrix}\\
G_{p+d}^{(2)}\triangleq&\begin{bmatrix}
    C \Phi_d \bar{L} &\cdots &C \Phi_d(A-\bar{L} \bar{C})^{p-1} \bar{L}
\end{bmatrix}.
\end{aligned}
\]
Then we have the following bound for the norm
\[
\begin{aligned}
    G_{p+d}G_{p+d}^{\tr}\leq& \left\|\bar{L}\right\|_2^2\left\|C\right\|_2^2\left\|\Phi_d\right\|_2^2\sum_{l=0}^{p-1}\left\|(A-\bar{L}\bar{C})^l\right\|_2^2\\&+\left\|C\right\|_2^2\sum_{l=1}^{d}\left\|L^{(l)}\right\|_2^2\left\|\Phi_{d-l}\right\|_2^2.
\end{aligned}\]
From Assumption \ref{asp: Diagonal} of $A-\bar{L}\bar{C}$, i.e., $\left\|(A-\bar{L}\bar{C})^p\right\|_2\leq M_1\rho(A-\bar{L}\bar{C})^p$, and Theorem \ref{thm: property} with uniform boundedness of $\Phi_l,\;l=1\dots,d$, we have 
$$
\left\| \lambda G_{p+d} V_{k}^{-\frac{1}{2}}\right\|_{2}^{2}\leq \frac{M_2}{1-\rho(A-\bar{L}\bar{C})^2}, \;\;\forall k\ge T_{\text{init}}.
$$
where $M_1$ and $M_2$ are constants only related to system parameters.
\subsection{Proof of Lemma \ref{lemma:regression-factor}: Uniform Boundedness of regression factor $\mathcal{E}_k$ with high probability}\label{sectionRegression}

In this subsection, we will provide a uniform upper bound of the regression error term  $\left\|\mathcal{R}_{k} \bar{Z}_{k}^{{\tr}} V_{k}^{-\frac{1}{2}}\right\|_{2}^{2}$ for all $k\ge T_{\text{init}}$.
Note that $r_k$ are mutually uncorrelated Gaussian random variable, i.e., $r_k\sim \mathcal{N}(0,CP^{(d+1)}C^\tr+R)$
 and we denote $\bar{R}_{(d)}=CP^{(d+1)}C^\tr+R$, then the matrix  
$\bar{R}_{(d)}^{-\frac{1}{2}}\mathcal{R}_{k}=\begin{bmatrix}
   \bar{R}_{(d)}^{-\frac{1}{2}}r_{p+d}&\cdots& \bar{R}_{(d)}^{-\frac{1}{2}}r_k
\end{bmatrix}$ is composed of mutually uncorrelated standard Gaussian random vectors. Then we consider the matrix $$\mathcal{E}_{k}\triangleq\bar{R}_{(d)}^{-\frac{1}{2}}\mathcal{R}_{k}\bar{Z}_{k}^{{\tr}} V_{k}^{-1}\bar{Z}_{k}\mathcal{R}_{k}^{{\tr}}\bar{R}_{(d)}^{-\frac{1}{2}},$$
and denote $V_{k}^{Z}=\bar{Z}_{k}^{{\tr}} V_{k}^{-1}\bar{Z}_{k}$. Then
with a similar technique with the proof of Lemma \ref{lm: propersistent}, we first choose
\[t=K^2\left(\frac{1}{2}\left\|V_{k}^{Z}\right\|_F^2+\frac{1}{c}\log\left(\frac{mk^2}{\delta}\right)\right),
\] 
where $K$ is still the $\psi_2$ norm of standard Gaussian random variable.
We can also verify
\[
\begin{aligned}
    \text{trace}(V_{k}^{Z})=&\text{trace}\left( \left(\lambda I+\bar{Z}_{k}\bar{Z}_{k}^{{\tr}}\right)^{-1}\bar{Z}_{k}\bar{Z}_{k}^{{\tr}}\right)\leq m_{p,d},\\
    \left\|V_{k}^{Z}\right\|_F^2=&\text{trace}\left(V_{k}^{Z}(V_{k}^{Z})^{\tr}\right)\leq m_{p,d},\; \left\|V_{k}^{Z}\right\|_2\leq1.\\
\end{aligned}
\]
Denote each row of $\bar{R}_{(d)}^{-\frac{1}{2}}\mathcal{R}_{k}$ as $g^{(i)}$, $i=1,\dots,m$, then with Lemma \ref{lm: HWinequality}, we also have
\[
\mathbb{P}\left\{g^{(i)}V_k^Zg^{(i)\tr} \geq \operatorname{tr}(V_{k}^Z)+t\right\} \leq \frac{2\delta}{mk^2},\;\; i=1,\dots,m.
\]
Then, with the fact that
$\left\|\mathcal{E}_{k}\right\|_2\leq \sum_{i=1}^m g^{(i)}V_k^Zg^{(i)\tr}$, we have
\[\mathbb{P}\left\{\left\|\mathcal{E}_{k}\right\|_2 \leq m\operatorname{trace}(V_{k}^Z)+mt\right\} \ge  1-\frac{2\delta}{k^2}.\]
Take a uniform bound over all $k$, and note that
\[m\operatorname{trace}(V_{k}^Z)+mt\leq \poly\left(d,\beta,\log \frac{1}{\delta}\right)\log k.\]
$p\leq \beta \log(k)$, Then we finally have
$$
\mathbb{P}\left\{\!\left\|\mathcal{E}_{k}\right\|_2 \!\leq\!\operatorname{poly}\!\left(\!d,\!\beta\!,\log\frac{1}{\delta}\!\right) \!\log k,\!\forall k\!\ge\! T_{\text{init}}\!\right\}\!\ge\! 1\!-\!\frac{\pi^2\delta}{3},
$$
and with the previous notation $\mathcal{E}_k\triangleq\left\|\mathcal{R}_{k} \bar{Z}_{k}^{{\tr}} V_{k}^{-\frac{1}{2}}\right\|_{2}^{2}$
$$
\mathbb{P}\!\left\{\!\mathcal{E}_k\!\leq\! \bar{\sigma}_{\bar{R}_{(d)}}\!\operatorname{poly}\!\Big(\!d,\!\beta,\!\log\!\frac{1}{\delta}\!\Big) \!\log k,\forall k\ge T_{\text{init}}\right\}\!\ge\! 1-\frac{\pi^2\delta}{3},
$$
where $\bar{\sigma}_{\bar{R}_{(d)}}$ is the largest singular value of $\bar{R}_{(d)}$.

\subsection{Proof of Lemma~\ref{lemma:accumulation-factor}: Uniform Boundedness of Accumulation factor $\mathcal{V}_N$ with high probability}\label{sec:proofAccumulation}
Without loss of generality, we assume $N=2^{N_E} T_{\text{init}}$, where $N_E$ is the number of epochs. Then
    we can decompose the $\mathcal{V}_{N}$ into 
\vspace{-5pt}
\[
\begin{aligned}
    \mathcal{V}_{N}\!\leq\!\! \max_{T_{\text{init}}\leq k\leq N}\left\|V_{k-1,p+d}^{-\frac{1}{2}}V_{k,p+d}^{\frac{1}{2}}\right\|_2^2\qquad\qquad\qquad\qquad\\ \times\sum_{l=1}^{N_E}\!\sum_{k=T_l}^{2T_l-2}\!\left\|V_{k, p_l+d}^{-\frac{1}{2}} Z_{k,p_l+d}\right\|_{2}^{2},
\end{aligned}
\vspace{-5pt}
\]
where the subscript of $p_l$ is to highlight the $p$ varies with epoch number $l$.
The inequality
$
\textstyle \sum_{k=T_{l}}^{2T_l-2}\left\|V_{k, p_l}^{-\frac{1}{2}} Z_{k,p_l}\right\|_{2}^{2}\leq \log\big(\det V_{2T_l-2,p_l}/\det V_{T_l-1,p_l}\big)
$ 
 directly holds with \cite[Lemma 1]{tsiamis9894660}.
In the following proof, we will also omit the subscripts $p$ and $d$ when unnecessary, in order to simplify notation. 

{\bf Uniform Boundedness of $\left\|V_{k}^{-\frac{1}{2}}V_{k+1}^{\frac{1}{2}}\right\|_2^2$ with high probability: }
Inspired by \cite{tsiamis9894660}, we still consider the successive representation of $Z_{k}$ generated by the minimal polynomial of $A$.
For the matrix $A$, suppose the minimal polynomial of $A$ takes the form
$
A^{d_0}= a_{d_0-1}A^{d_0-1}+\dots+a_0,
$
where $d_0$ is the dimension of the minimal polynomial of $A$.
Then from the Lemma 2 of \cite{tsiamis9894660}, we have the following successive representation of the sample $Z_{k}$
\vspace{-5pt}
$$Z_{k}=a_{d_0-1} Z_{k-1}+\ldots+a_{0} Z_{k-d_0}+\Delta_k.\vspace{-5pt}$$
The term $\Delta_k=\left[\Delta_{1,k}^\tr,\Delta_{2,k}^{\tr}\right]^\tr$, where
\vspace{-5pt}
\[
\begin{aligned}
    \Delta_{1,k}=&\sum_{s=0}^{d_0} \operatorname{diag}\Big(\underbrace{L_{s}, \ldots, L_{s}}_{d}\Big) \tilde{E}_{k-s,d}^{(1)}\\
    \Delta_{2,k}=&\sum_{s=0}^{d_0} \operatorname{diag}\Big(\underbrace{\bar{L}_{s}, \ldots, \bar{L}_{s}}_p\Big) \tilde{E}_{k-s,p}^{(2)},
\end{aligned}
\vspace{-5pt}
\]
and the parameters take the form of
$
    L_{s}=-a_{d_0-s} I_{m}-\sum_{t=1}^{s-1} a_{d_0-s+t} C A^{t-1} L+C A^{s-1} L,
    \bar{L}_{s}=-a_{d_0-s} I_{m}-\sum_{t=1}^{s-1} a_{d_0-s+t} \bar{C} A^{t-1} \bar{L}+\bar{C} A^{s-1} \bar{L},
$ respectively. The augmented noise terms take the form of 
$
    \tilde{E}_{k-s, d}^{(1)}=[e_{k-s-1}^{(1)\tr},\dots,e_{k-s-d}^{(1)\tr}]^\tr,
    \tilde{E}_{k-s, p}^{(2)}=[e_{k-s-1}^{(2)\tr},\dots,e_{k-s-p}^{(2)\tr}]^\tr.
$
The random variable satisfies $e_{k}^{(1)}\sim\mathcal{N}\left(0,CPC^\tr+R\right)$ and $e_{k}^{(2)}\sim\mathcal{N}\left(0,\bar{C}\bar{P}\bar{C}^\tr+\bar{R}\right)$, respectively. 

For $\Delta_{1,k}$, as $d_0$ is fixed for a given $A$, we first have $\left\|L_{s}\right\|_{2}\leq M$ holds for all $s\leq d_0$, then we have the bound
\vspace{-5pt}
$$
\left\|\Delta_{1,k}\right\|_{2}^2 \leq (d_0+1)^2 \max _{0 \leq s \leq d_0}\left\|L_{s}\right\|_{2}^2\;  d\sup_{0\leq s\leq d+d_0}\big\|e_{k-s}^{(1)}\big\|_{2}^2.
\vspace{-5pt}
$$
With Lemma \ref{lm: Deviation}, we have the event 
\vspace{-5pt}
\[\mathcal{E}_{e^{(1)}}\triangleq\left\{\big\|e^{(1)}_k\big\|_{2}^2 \leq \bar{\sigma}^{(1)}(2m+3\log \frac{k^2}{\delta}), \;\; \forall k\ge 1\right\}
\vspace{-5pt}
    \]
    holds with $\mathbb{P}\left(\mathcal{E}_{e^{(1)}}\right)\ge 1-\frac{\pi^2 \delta}{6}$, where $\bar{\sigma}^{(1)}$ is the largest eigenvalue of $CPC^\tr+R$. Therefore conditioned on $\mathcal{E}_{e^{(1)}}$, we have
    $\mathbb{P}\big(\left\|\Delta_{1,k}\right\|_{2}^2\leq \poly\left(d_0,d,\log\frac{1}{\delta}\right)\log k,\forall k\ge T_{\text{init}}\big)\ge 1-\frac{\pi^2\delta}{6}.
    $
With a similar process and $p\leq\beta\log k$, we have
$
\mathbb{P}\big(\left\|\Delta_{2,k}\right\|_{2}^2\leq \poly\left(d_0,\beta,\log\frac{1}{\delta}\right)\log^2 k,\forall k\ge T_{\text{init}}\big)\ge 1-\frac{\pi^2\delta}{6}.
$
Then we consider the term $\left\|V_{k}^{-\frac{1}{2}}V_{k+1}^{\frac{1}{2}}\right\|_2^2$, together with the structure of $V_{k}$, we have
\vspace{-5pt}
$$
\begin{aligned}
    &\left\|V_{k}^{-\frac{1}{2}}V_{k+1}^{\frac{1}{2}}\right\|_2^2=1+Z_{k+1}V_{k}^{-1}Z_{k+1} \\
    &\quad\leq 1\!+\!2\Delta_{k+1}^{\tr} V_{k}^{-1}\Delta_{k+1}\!+\!2d_0\left\|a\right\|_2^2\sum_{s = 0}^{d_0-1}Z_{k-s}^{\tr} V_{k}^{-1}Z_{k-s}.
\end{aligned}
\vspace{-5pt}
$$
Conditioned on $\mathcal{E}_{e^{(1)}}$ and $\mathcal{E}_{e^{(2)}}$, together with Lemma~\ref{lm: propersistent}, we have the following bound
\vspace{-5pt}
$$
\begin{aligned}
    \Delta_{k+1}^{\tr} V_{k}^{-1}\!\Delta_{k+1}\!\leq\! \frac{\left\|\Delta_{k+1}\right\|^2}{\frac{\sigma_{\bar{R}}}{4}k}
\!\leq \!\frac{\poly\left(d_0,d,\beta,\log\frac{1}{\delta}\right)\!\log^2 k}{k}.
\end{aligned}
\vspace{-5pt}
$$
The second equality holds due to the fact that $\left\|\Delta_{k+1}\right\|^2=\left\|\Delta_{1,k+1}\right\|^2+\left\|\Delta_{2,k+1}\right\|^2$.
Note that 
$
\frac{\log^2 k}{k}\leq \frac{4}{e^2},\;\; \forall k\ge 1.
$
Consider the term $d_0$ as a constant, since the system model is fixed, then we have
\vspace{-5pt}
$$
\mathbb{P}\Big(\!\Delta_{k+1}^{\tr} \!V_{k}^{-1}\!\!\Delta_{k+1}\!\!\leq\! \poly\!\big(d,\!\beta,\!\log\!\frac{1}{\delta}\big),\!\forall k\!\ge \!T_{\textnormal{init}}\!\Big)\!\ge \!1\!-\frac{\pi^2 \delta}{2}.
\vspace{-5pt}
$$
Then for the quadratic form $\sum_{s = 0}^{d_0-1}Z_{k-s}^{\tr} V_{k}^{-1}Z_{k-s}$, from Woodbury Equality, there is
\vspace{-5pt}
$$
\begin{aligned}
    Z_{k}^{\tr}\left(V_{k-1}+Z_{k} Z_{k}^{{\tr}}\right)^{-1}Z_{k}=&\frac{Z_{k}^{{\tr}}V_{k-1}^{-1}Z_{k}}{1+Z_{k}^{{\tr}}V_{k-1}^{-1}Z_{k}}< 1.
\end{aligned}
\vspace{-5pt}
$$
With a similar technique, it is easy to verify that 
\vspace{-5pt}
$$
2d_0\left\|a\right\|_2^2\sum_{s = 0}^{d_0-1}Z_{k-s}^{\tr} V_{k}^{-1}Z_{k-s}\leq 2d_0^2\left\|a\right\|_2^2.
\vspace{-5pt}
$$
Then we can find a uniform {\bf constant} bound $M=\operatorname{poly}\left(d_0,d,\left\|a\right\|_2,\beta,\operatorname{log}\frac{1}{\delta}\right)$ for the term $\big\|V_{k}^{-\frac{1}{2}}V_{k+1}^{\frac{1}{2}}\big\|_2^2$ with high probability $1-\frac{\pi^2}{2}\delta$, which does not scale with the time step $k$.

\vspace{-5pt}
{\bf Uniform Boundedness of $\sum_{k=T_{\text{init}}}^{N}\left\|V_{k}^{-\frac{1}{2}} Z_{k}\right\|_{2}^{2}$ with high probability: }
The analysis of this part takes a similar procedure as that in \cite[Appendix C.7]{qian2025model}, where the persistent excitation of Gram matrix $V_{k,p}$ guarantees that the sum $\sum_{k=T_{\text{init}}}^{N}\big\|V_{k}^{-\frac{1}{2}} Z_{k}\big\|_{2}^{2}$ is dominated by the accumulation in the last epoch, i.e., $\sum_{k=T_{\text{init}}}^{N}\big\|V_{k}^{-\frac{1}{2}} Z_{k}\big\|_2^2=\mathcal{O}\big(\sum_{k=N/2}^{N}\big\|V_{k}^{-\frac{1}{2}} Z_{k}\big\|_{2}^{2}\big)$. Then the doubling trick does not yield additional regret in terms of $\log N$.

With Lemma 1 in \cite{tsiamis9894660}, we first have the accumulation error of two successive epochs, i.e.,
$$
\sum_{k=T_{l}}^{4T_l-4}\!\!\Big\|V_{k}^{-\frac{1}{2}} \!Z_{k}\Big\|_{2}^{2}\!\!\!\leq\! \log \!\frac{\det(V_{4T_l\!-\!4,p^\prime \!+d})}{\det(V_{2T_l\!-\!2,p^\prime \!+d})}+\log\!\!\frac{\det(V_{2T_l-2,p+d})}{\det(V_{T_l-1,p+d})},
$$
where $p'=\beta \log(2T_l-2)$, $p=\beta \log(T_l-1)$ and $T_l$ is the length of the $l$-th epoch. 
We have $p'-p = \beta \log 2$, and we can separate the augmented samples $Z_{k,p^\prime+d}$ as
\vspace{-5pt}
$$
Z_{k,p'+d} = \begin{bmatrix}
   Z_{k, p+d}^{\tr}&Y_{k-p^\prime-d:k-p-d-1}^{c \tr}
\end{bmatrix}^{\tr},
\vspace{-8pt}
$$
where
$
Y_{k-p^\prime-d:k-p-d-1}^{c}=\left[y_{k-p-d-1}^{c\tr},\dots,y_{k-p^\prime-d}^{c\tr}\right]^\tr
$
is a combination of a sequence of global observations arising from different values of 
$p$ across different epochs.
Then the Gram matrix $V_{2T_l-2,p^{\prime}+d}$ can be correspondingly partitioned as 
$
\begin{aligned}
V_{2T_l-1,p^\prime+d}=\lambda I+ \begin{bmatrix}
Z_{11}&Z_{12}\\Z_{12}^{\tr}&Z_{22}
\end{bmatrix},
\end{aligned} 
$
where
$
Z_{11}=V_{2T_l-2,p+d}-\sum_{k=p+d}^{p^\prime+d-1}Z_{k,p+d}Z_{k,p+d}^\tr.
$
Then we have $
    \det V_{2T_l-1,p'+d}=
\det(\lambda I + {Z}_{22}-{Z}_{12}^\tr(\lambda I + {Z}_{11})^{-1}{Z}_{12})
\cdot\det(\lambda I + {Z}_{11})
$
By Schur complement lemma and Lemma~\ref{lm: propersistent}, we have
\vspace{-5pt}
\[
{Z}_{22}-{Z}_{12}^\tr(\lambda I + {Z}_{11})^{-1}{Z}_{12}\ge {Z}_{22}-{Z}_{12}^\tr{Z}_{11}^{-1}{Z}_{12}\ge 0,
\vspace{-5pt}
\]
where the persistent excitation of the Gram matrix guarantees that $Z_{11}$ is invertible.
Then we have
\vspace{-5pt}
\[\det(\lambda I + {Z}_{22}-{Z}_{12}^\tr(\lambda I + {Z}_{11})^{-1}{Z}_{12})\ge \lambda^{\bar{m}(p^\prime-p)}.
\vspace{-5pt}\]
Consider the relationship between $Z_{11}$ and $V_{2T_l-2,p+d}$,
with Lemma \ref{lm: Deviation}, we have 
$
\mathbb{P}\Big(\left\|\left(\Gamma_{k, p+d}^Z\right)^{-1 / 2} Z_{k, p+d}\right\|_{2}^2 \leq 2m_{p,d}+3 \log \frac{k^2}{\delta},\;\forall k\ge T_{\textnormal{init}}\Big)\ge 1-\frac{\pi^2\delta}{6}$
where $m_{p,d}\triangleq\bar{m}p+md$ is the dimension of the sample $Z_{k,p+d}$ for given $p$ and $d$.
With the above analysis, we further obtain
\vspace{-5pt}
$$
\begin{aligned}
    Z_{k, p+d}Z_{k, p+d}^{\tr} \leq M\left(2m_{p,d}+3 \log \frac{k^2}{\delta}\right)(p+d) k^{2\kappa-1}I.\\
\end{aligned}
\vspace{-5pt}
$$
The above inequality holds with $\mathbb{E}\left\{Z_{k, p+d}^TZ_{k, p+d}\right\} \leq$
We have
\vspace{-5pt}
$$
\begin{aligned}
    \sum_{k=p+d}^{p^\prime+d-1}\!\!Z_{k,p+d}Z_{k,p+d}^\tr\!\leq& M(p^\prime-p)\!(p^\prime +d)^{2\kappa}\\
    \vspace{-5pt}&\times \Big(2m_{p,d}+\!3 \log (p^\prime + d)^2/\delta\Big)I.
\end{aligned}
\vspace{-5pt}
$$
Note that $p^\prime=\beta\log(2T_l-2)$, then we have
\[
\sum_{k=p+d}^{p^\prime+d-1}Z_{k,p+d}Z_{k,p+d}^\tr\leq\poly(d,\beta,\log\frac{1}{\delta})\log^{2\kappa+2}(2T_l-2).
\]
With Lemma~\ref{lm: propersistent}, we further have
$$
\begin{aligned}
    \sum_{k=p+d}^{p^\prime+d-1}\!\!\! Z_{k, p+d}Z_{k, p+d}^{\tr} \!\leq\! \poly(d,\!\beta,\!\log\frac{1}{\delta})\!\frac{\log^{2\kappa+2}(2T_l\!-\!2)}{\sigma_{\bar{R}}(2T_l\!-\!2)}\!Z_{11}.
\end{aligned}
$$
The above inequality holds due to 
$\lambda_{\min}(Z_{11})\ge \lambda_{\min}(V_{2T_l-2,p^\prime+d})\ge \sigma_{\bar{R}}(2T_l-2)/4$.
Note that there exists a uniform $M_1$, such that
$$
M_1\!\ge\!\frac{\poly(d,\beta,\log\frac{1}{\delta})\log^{2\kappa+2}(2T_l-2)}{\sigma_{\bar{R}}(2T_l-2)},\,\forall l=1,\dots,\!N_E,
$$
and we can verify that
$
M_1=\poly\left(d,\beta,\log\left(1/\delta\right)\right).$
Then we have
$V_{2T_l-2,p+d}\leq (1+M_1)(\lambda I +Z_{11}),$
and
$$
\log \det V_{2T_l-2,p+d}\leq m_{p,d}\log(1+M_1)+\log\det (\lambda I +Z_{11}).
$$
Finally, we have 
\[
\log\det V_{2T_l-2,p^\prime+d}\ge \log\det(\lambda I +Z_{11})+\bar{m}(p^\prime-p)\log \lambda,
\]
and correspondingly 
$$
\begin{aligned}
    \log\det V_{2T_l-2,p+d}-&\log\det V_{2T_l-2,p^\prime+d}\\
    \leq& m_{p,d}\log(1+M_1)-\bar{m}(p^{\prime}-p)\log(\lambda)
\end{aligned}
$$
holds for all $l=1,\dots,N_E$ with probability at least $1-\frac{\pi^2\delta}{3}$.
We denote $p_l=\beta\log(2^{l-1}T_{\text{init}})$, then we have
$$
\begin{aligned} &\sum_{k=T_{\text{init}}}^{N}\left\|V_{k,p+d}^{-\frac{1}{2}} Z_{k, p+d}\right\|_{2}^{2}=\sum_{l=1}^{N_E}\log\frac{\det(V_{2T_l-2,p_l+d})}{\det(V_{T_l-1,p_l+d})}\\
\leq& \log\det(V_{2T_{N_E}-2,p_{N_E+d}})+\sum_{l=1}^{N_E}(\bar{m}p_l+md)\log(1+M_1)\\&-\log\det(V_{T_{1}-1,p_{1}})-\bar{m}N_E\beta\log(2)\log(\lambda)\\
\leq &\log\det(V_{N,p_{N_E}})+N_E \log(1+M_1)(\bar{m}\beta\log(N)+md).
\end{aligned}
$$
Note that $N=2^{N_E}T_{\text{init}}$, we have $N_E=O(\log N)$. With Lemma \ref{lm: upperbound}, we also have
\[
V_{N,p_{N_E}}\leq \operatorname{poly}\left(d,\beta,\log 1/\delta\right)N^{2\kappa
+1}I,
\]
and
$$
\begin{aligned}
    &\log\det(V_{N,p_{N_E}})\leq \log\det\Big(\operatorname{poly}\big(d,\beta,\log 1/\delta\big)N^{2\kappa
+1}I\Big)\\&\qquad\quad\leq (\bar{m}p_{N_E}+md)\log\Big(\operatorname{poly}\big(d,\beta,\log 1/\delta\big)N^{2\kappa
+1}\Big).
\end{aligned}
$$
Hence we have 
$$
\sum_{k=T_{\text{init}}}^{N}\left\|V_{k,p+d}^{-\frac{1}{2}} Z_{k, p+d}\right\|_{2}^{2}\leq\operatorname{poly}\left(d,\beta,\log\frac{1}{\delta}\right)\log^2(N),
$$
holds for all $N$ with probability at least $1-\frac{3\pi^2\delta}{6}$.

\subsection{Uniform Boundedness of $\left(\hat{y}_k-\tilde{y}_k\right)^{\tr} r_k$}
In this subsection, in order to simplify the proof, we directly utilize the result in Theorem 1 and Theorem 3 in \cite{tsiamis9894660}, and we have
$$
\begin{aligned}
   \sum_{k=T_{\text{init}}}^{N}r_k^\top(\hat{y}_k-\tilde{y}_k)=&\operatorname{poly}\left(\frac{1}{\delta}\right)\sqrt{\mathcal{L}_N}\log(\mathcal{L}_N)\\=&\operatorname{poly}\left(\frac{1}{\delta}\right)o(\mathcal{L}_N) 
\end{aligned}
$$
with probability $1-\delta$.   
\end{document}